\documentclass[12pt]{article}
\usepackage{amsmath,amsfonts,amssymb,amsthm,array}
\usepackage{commath}
\usepackage{bm,color,setspace,graphicx,times}
\usepackage{wrapfig,subcaption}
\usepackage{algorithm}
\usepackage[noend]{algpseudocode}
\usepackage{mathtools}
\usepackage[margin=1in]{geometry}
\graphicspath{{Figs/}}

\theoremstyle{plain}
\newtheorem{thm}{Theorem}[section]
\newtheorem{lem}[thm]{Lemma}
\newtheorem{prop}[thm]{Proposition}
\newtheorem{cor}[thm]{Corollary}
\newtheorem{defn}[thm]{Definition}
\newtheorem{claim}[thm]{Claim}
\newtheorem{asmn}[thm]{Assumption}
\newtheorem{con}[thm]{Conjecture}
\newtheorem{rem}[thm]{Remark}

\usepackage{qtree}

\usepackage{hyperref}

\usepackage[nameinlink,noabbrev,capitalize]{cleveref} 
\crefalias{subequation}{equation}
\crefalias{thm}{theorem}

\let\oldlemma\lem
\renewcommand{\lem}{%
  \crefalias{thm}{lem}% Theorem counter now looks like Lemma
  \oldlemma}
\Crefname{lem}{Lemma}{Lemmas}

\let\olddefn\defn
\renewcommand{\defn}{%
  \crefalias{thm}{defn}% Theorem counter now looks like Definition
  \olddefn}
\Crefname{defn}{Definition}{Definitions}

\let\oldrem\rem
\renewcommand{\rem}{%
  \crefalias{thm}{rem}% Theorem counter now looks like Remark
  \oldrem}
\Crefname{rem}{Remark}{Remarks}

\let\oldcor\cor
\renewcommand{\cor}{%
  \crefalias{thm}{cor}% Theorem counter now looks like Corollary
  \oldcor}
\Crefname{cor}{Corollary}{Corollaries}

\let\oldclaim\claim
\renewcommand{\claim}{%
  \crefalias{thm}{claim}% Theorem counter now looks like Claim
  \oldclaim}
\Crefname{claim}{Claim}{Claims}

\let\oldprop\prop
\renewcommand{\prop}{%
  \crefalias{thm}{prop}% Theorem counter now looks like Prop
  \oldprop}
\Crefname{prop}{Proposition}{Propositions}

\let\oldcon\con
\renewcommand{\con}{%
  \crefalias{thm}{con}% Theorem counter now looks like Con
  \oldcon}
\Crefname{con}{Conjecture}{Conjectures}

\let\oldasmn\asmn
\renewcommand{\asmn}{%
  \crefalias{thm}{asmn}% Theorem counter now looks like Asmn
  \oldasmn}
\Crefname{asmn}{Assumption}{Assumptions}

\usepackage[normalem]{ulem}

\definecolor{darkgrn}{rgb}{0, 0.8, 0}

\newcommand{\R}{ {\mathbb R} }

\newcommand{\hK}{\hat{K}}
\newcommand{\tK}{\tilde{K}}
\newcommand{\hV}{\hat{V}}
\newcommand{\hD}{\hat{D}}
\newcommand{\hE}{\hat{E}}
\newcommand{\hF}{\hat{F}}
\newcommand{\hT}{\hat{T}}
\newcommand{\hG}{\hat{G}}
\newcommand{\hu}{\hat{u}}
\newcommand{\hv}{\hat{v}}
\newcommand{\tv}{\tilde{v}}
\newcommand{\he}{\hat{e}}
\newcommand{\hf}{\hat{f}}
\newcommand{\htt}{\hat{t}}
\newcommand{\ttt}{\tilde{t}}
\newcommand{\hrho}{\hat{\rho}}
\newcommand{\CH}{\mathcal{C}^H}
\newcommand{\CO}{\mathcal{C}^O}
\newcommand{\hCH}{\hat{\mathcal{C}}^H}
\newcommand{\hCO}{\hat{\mathcal{C}}^O}
\newcommand{\sT}{\mathcal{T}}
\newcommand{\cP}{\mathcal{P}}
\newcommand{\tf}{\tilde{f}}

\newcommand{\diam}{\operatorname{D}}
\newcommand{\per}{\operatorname{P}}
\newcommand{\area}{\operatorname{A}}
\newcommand{\g}{\gamma}
\newcommand{\gl}{\lambda} % Greek l
\newcommand{\ghf}{\gamma(\hf)}

\title{Euler Transformation of Polyhedral Complexes}

\author{
    Prashant Gupta \hspace*{1in}   Bala Krishnamoorthy \\
    Department of Mathematics and Statistics, Washington State University \\
    \{prashant.gupta,kbala\}@wsu.edu
}

\begin{document}

\maketitle

\begin{abstract}

  We propose an \emph{Euler transformation} that transforms a given $d$-dimensional cell complex $K$ for $d=2,3$ into a new $d$-complex $\hK$ in which every vertex is part of the same even number of edges.
Hence every vertex in the graph $\hG$ that is the $1$-skeleton of $\hK$ has an even degree, which makes $\hG$ Eulerian, i.e., it is guaranteed to contain an Eulerian tour.
Meshes whose edges admit Eulerian tours are crucial in coverage problems arising in several applications including 3D printing and robotics.

For $2$-complexes in $\R^2$ ($d=2$) under mild assumptions (that no two adjacent edges of a $2$-cell in $K$ are boundary edges), we show that the Euler transformed $2$-complex $\hK$ has a geometric realization in $\R^2$, and that each vertex in its $1$-skeleton has degree $4$.
We bound the numbers of vertices, edges, and $2$-cells in $\hK$ as small scalar multiples of the corresponding numbers in $K$.

We prove corresponding results for $3$-complexes in $\R^3$ under an additional assumption that the degree of a vertex in each $3$-cell containing it is $3$.
In this setting, every vertex in $\hG$ is shown to have a degree of $6$.

  We also present bounds on parameters measuring geometric quality (aspect ratios, minimum edge length, and maximum angle of cells) of $\hK$ in terms of the corresponding parameters of $K$ for $d=2,3$.
Finally, we illustrate a direct application of the proposed Euler transformation in additive manufacturing.
\end{abstract}

\section{Introduction} \label{sec:intro}

An \emph{Eulerian circuit}, or Eulerian tour, in a finite graph $G=(V,E)$ is a closed walk that traverses every edge in $E$ exactly once.
In other words, the walk starts and ends at the same node while possibly visiting some nodes in $V$ multiple times while \emph{covering} all edges.
The classical result attributed to Euler \cite{Eu1736,Hi1873} states that $G$ has an Eulerian tour if and only if $G$ is connected and every node in $V$ has an even degree.
We are interested in Eulerian circuits in the context of coverage problems arising in additive manufacturing (3D printing), robotics, and other areas.
The domain to be covered is usually modeled by a cell complex, i.e., a tessellation.
Triangulations or hexagonal meshes in 2D, and cubical meshes or tetrahedralizations in 3D are typical examples.
Complete coverage of the domain is ensured by traversing all edges (and hence all vertices) in the mesh.
Efficient traversal of all edges in a contiguous fashion becomes critical in this context.

In additive manufacturing, we first print the outer ``shell'' or boundary of the 3D object in each layer.
We then cover the interior space by printing an \emph{infill lattice} \cite{BrBrWiHa2012,WuAaWeSi2018}, which is typically a standard mesh where any two edges meet at most at a vertex.
In large scale additive manufacturing,  printing most, if not all, edges of the infill lattice in a contiguous manner is critical to decrease non-print motions of the printer-head.
The problem of coverage path planning in robotics seeks to find a path that passes through all points while avoiding obstacles \cite{GaCa2013}.
Standard approaches for such coverage problems employ graph-based algorithms \cite{Xu2011}.
A robot is typically required to cover all vertices and edges of the graph, while using the edges sequentially without repetition \cite{CaHuHa1988}.
Traversing the edges along an Eulerian tour is required to address these challenges.
But the graph made of the vertices and edges in a cell complex is not always guaranteed to contain an Eulerian tour.
On the other hand, any cell complex that guarantees the existence of an Eulerian tour covering its edges would be expected to also offer reasonable bounds on the quality of elements, e.g., as measured by the aspect ratios of its cells.

\subsection{Our contributions} \label{ssec:contrib}

We propose a method that transforms a given $d$-dimensional cell complex $K$ (or $d$-complex, in short) for $d=2,3$ into a new $d$-complex $\hK$ in which every vertex is part of the same even number of edges.
Hence every vertex in the graph $\hG$ that is the $1$-skeleton of $\hK$ has an even degree, which makes $\hG$ Eulerian, i.e., it is guaranteed to contain an Eulerian tour.
We refer to this method as an \emph{Euler transformation} of a polyhedral mesh (or cell complex).
We first describe the Euler transformation of a $d$-complex for $d=2,3$ \emph{abstractly}, i.e., without specifying details of a geometric realization.

For $2$-complexes in $\R^2$ ($d=2$) under mild assumptions (that no two adjacent edges of a $2$-cell in $K$ are boundary edges), we show that the Euler transformed $2$-complex $\hK$ has a geometric realization in $\R^2$, and that each vertex in its $1$-skeleton has degree $4$.
We bound the numbers of vertices, edges, and polygons in $\hK$ as small scalar multiples of the corresponding numbers in $K$.
We prove corresponding results for $3$-complexes in $\R^3$ under an additional assumption that each vertex in $K$ is connected to three edges in a $3$-cell that contains the vertex, i.e., the degree of each vertex in the $1$-skeleton of each $3$-cell in $K$ containing the vertex is $3$.
In this setting, every vertex in $\hG$ is shown to have a degree of $6$.
We show another nice geometric property of the Euler transformed $3$-complex: every edge in $\hK$ is shared by exactly four polygons (i.e., faces).
As a result, if we slice $\hK$ by a plane that cuts across only edges but does not intersect any vertices  in $\hK$, the resulting $2$-complex is guaranteed to be Euler, with each vertex having degree $4$.

Next, we present bounds on parameters measuring geometric quality (aspect ratios) of $\hK$ in terms of the corresponding parameters of $K$ (for $d=2, 3$).
One can control these quality measures by choosing user-defined offset parameters appropriately.
Finally, we illustrate a direct application of the proposed Euler transformation in additive manufacturing.
We present a complete algorithmic framework for continuous toolpath planning in additive manufacturing using our Euler transformation in a separate paper \cite{GuKrDr2020}.

\medskip
We illustrate the Euler transformation for $d=2$ in Figure \ref{fig:eulermeshillust}.
Given a $2$-dimensional cell complex $K$ tessellating a rectangular region in $\R^2$, the Euler transformation produces the $2$-complex $\hK$ tessellating the same region with every vertex having degree $4$.

\subsection{Related Work} \label{ssec:relwork}
\vspace*{-0.05in}

In one of the earliest works on degree-constrained triangulations, Jansen \cite{Ja1993} proved that it is NP-complete to decide whether a plane geometric graph can be triangulated with degree at most $7$.
Hoffmann and Kriegel showed that a $2$-connected, bipartite, planar graph can be triangulated such that the resulting graph is $3$-colorable \cite{HoKr1996}, implying all vertices have even degrees.
Aichholzer et al.~studied plane graphs with parity constraints on the vertices \cite{AiHaHoPiRoSpVo2009,AiHaHoPiRoSpVo2014}, and showed that we can always find a plane tree, two-connected outerplanar graph, or a pointed pseudo-triangulation that satisfies all but at most three parity constraints.
For triangulations, they showed that about $2/3$ of the constraints can be satisfied. But in the worst case, there is a linear number of constraints that cannot be fulfilled.
Pel\'aez et al.~\cite{PeRaUr2010} improved the lower bound on the number of even degree vertices in a triangulation to around $4/5$ of the total number.
Aichholzer et al.~\cite{AiHaHoPiRoSpVo2014} showed that it is NP-complete to decide whether there exists a triangulation of a simple polygon with polygonal holes that satisfies all parity constraints.
Recently, Gewali and Gurung \cite{GeGu2018} have proposed a heuristic algorithm for triangulating a planar annular region with increased number of even degree vertices.

Alvarez \cite{Al2015} studied parity-constrained triangulations with Steiner points.
For a given set of points $P$, Alvarez showed how to construct a set of Steiner points $S$ such that a triangulation of $P \cup S$ can be always constructed such that all vertices in the triangulation are even (or odd).
At the same time, one might have to choose two of these Steiner points \emph{outside} the convex hull of $P$.
Further, this result does not apply to input polygons with polygonal holes.

We consider polyhedral complexes, which are more general than simplicial complexes, and are increasingly used in computational mathematics \cite{FlGiSu2014,GiRaBa2012,MuWaWa2014} and in robotics \cite{GaCa2013,GaMoAbMa2006}.
At the same time, degree constrained polyhedral complexes have not received much attention.

Comparisons of our Euler transformation to other subdivision schemes such as Catmull-Clark \cite{CatmullClark1978} and Doo-Sabin \cite{DoSa1978} were presented in a separate paper \cite{GuKrDr2020}.
While not related, Edelsbrunner's work on deformable smooth surfaces \cite{Ed1999} indicated some coincidental similarities to our work.
In particular, all interior vertices of the mixed cells graph shown in Figure 10 of this paper \cite{Ed1999} have degree 4, just as in the case of our \cref{fig:eulermeshillust}.
\begin{figure}[hbp!] 
    \centering
    \includegraphics[width=70mm,height=70mm]{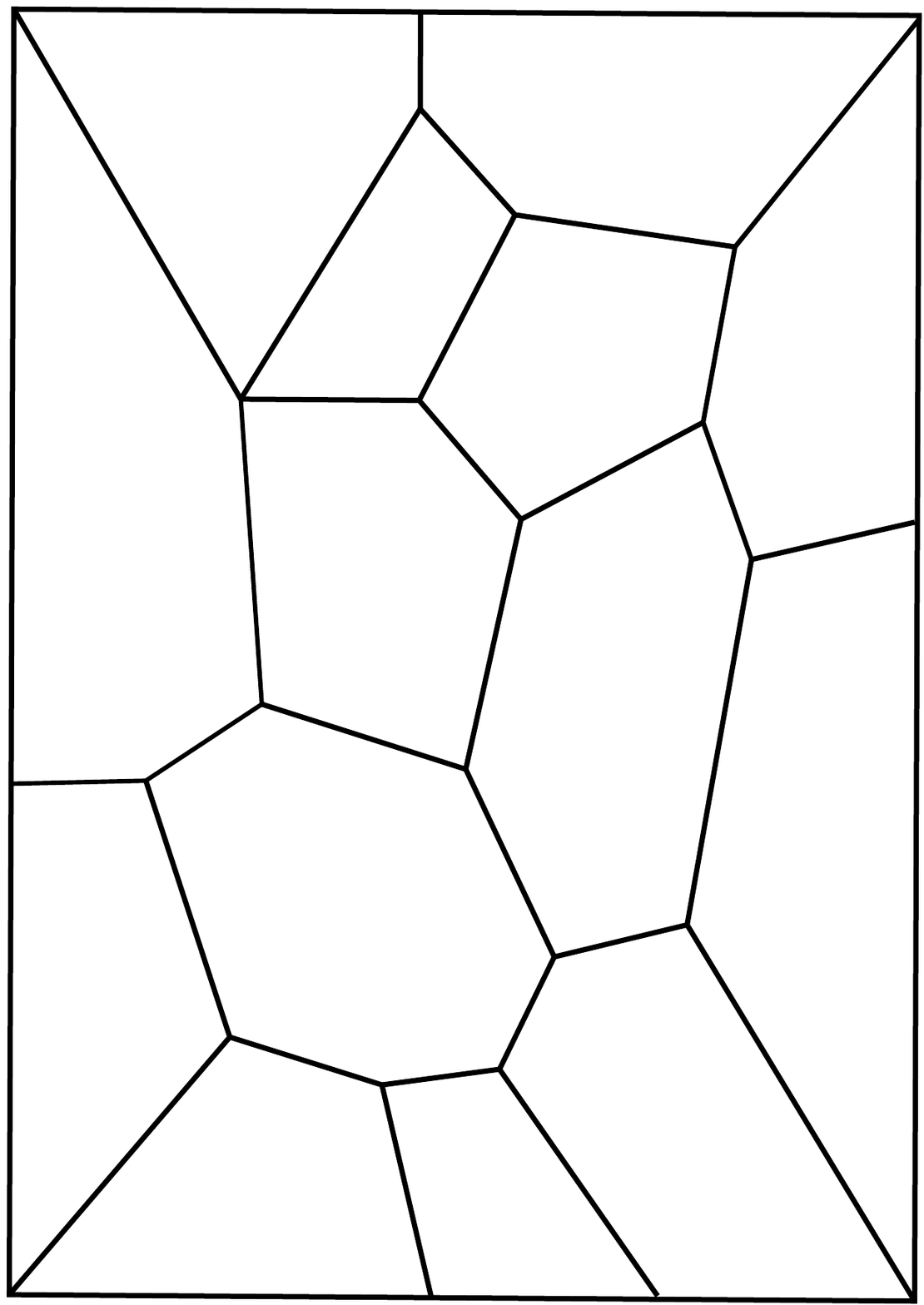}
    \quad\quad
    \includegraphics[width=70mm,height=70mm]{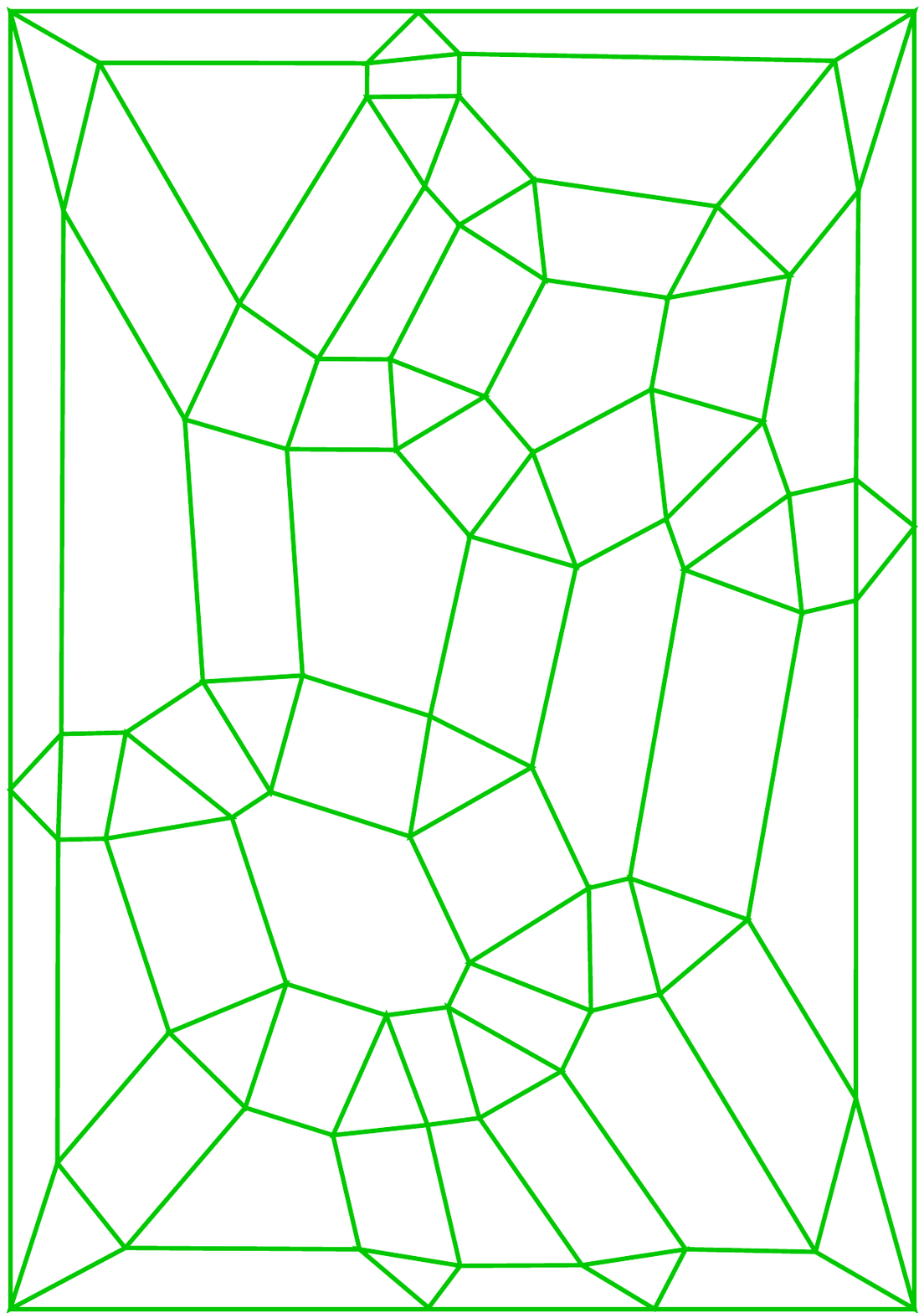}
    \caption{
      A $2$-complex $K$ in the plane (left) and its Euler transformation $\hK$ (right).
      Every vertex in $\hK$ has degree $4$.
    }
    \label{fig:eulermeshillust}
\end{figure}

%\vspace*{-0.1in}
\section{Notations, Definitions on Polyhedral Complexes} \label{sec:defns}

We present definitions that we use to specify properties of the input complex $K$ as well as the Euler transformed complex $\hK$.
See standard books on algebraic topology \cite{Hatcher2002,Munkres1984} for details.
For quick reference, we collect important notation used throughout the paper in Table \ref{tab:notation}.

\begin{table}[h!]
  \centering
  \caption{\label{tab:notation} Notations used in the paper, and their explanations.}
  \begin{tabular}{ll}
    \hline
    Notation & definition/interpretation \\    \hline
    \vspace*{-0.1in} \\
    $K, \hK$     & input complex and its transformed complex \\
    $|K|$        & underlying space of complex $K$\\
    $\CO, \hCO$  & single $d$-cell defining the \emph{outside}, before and after transformation\\
    $\CH, \hCH$  & collection of $d$-cells defining holes, before and after transformation \\
    $v_i, e_i, f_i, t_i$ & (generic) vertex, edge, polygon ($2$-cell), and $3$-cell in $K$\\
    $\hv_i, \he_i, \hf_i, \htt_i$ & (generic) cells of dimensions 0--3 in $\hK$\\
    $\hf, \hf_e, \hf_v$ & Class-$1$, $2$, $3$ cells in $\hK$ in 2D \\
    $\htt, \htt_f, \htt_e, \htt_v$ & Class-$1$, $2$, $3$, $4$ cells in $\hK$ in 3D\\
    $V, E, F, T$ & set of vertices, edges, polygon ($2$-cells), and $3$-cells in $K$ \\
    $\hV, \hE, \hF, \hT$ & sets of vertices, edges, polygon ($2$-cells), and $3$-cells in $\hK$\\
    $\rho$, $D$ & inradius and diameter of a cell in $K$\\
    $\hrho$, $\hD$ & inradius and diameter of a cell in $\hK$\\ 
    $\gamma(f), \gamma(\hf)$ & aspect ratios of cells $f \in K, \hf \in \hK$ \\
    $\lambda$, $\mu$ & aspect ratio parameters for a particular edge $\hK$ \\
    $\lambda^*$, $\mu^*$ & parameters over all edges in $\hK$ \\
    $|e_{\min}|, |e_{\max}|$ & minimum and maximum edge length in a $2$-cell in $K$ \\
    $|\he_{\min}|, |\he_{\max}|$ & minimum and maximum edge length in a $2$-cell in $\hK$\\
    $\theta_{\min}$ & minimum interior angle of a $2$-cell in $K$ \\
    $\alpha, \beta$ & maximum and minimum angles formed by edges of $\hf_v$ on vertex $v$\\
    \hline
  \end{tabular}
\end{table}

\begin{defn} \label{def:plyhdcplx}
  \emph{\bfseries (Polyhedral complex)}
  A polyhedral complex (also called polytopal complex) $K$ is a collection of \emph{polyhedra} (polytopes) in some Euclidean space $\R^d$ such that every face of a polyhedron in $K$ is also included in $K$, and the nonempty intersection of any two polyhedra in $K$ is a face of both.
  The polyhedra in $K$ are referred to as its \emph{cells}.
  The \emph{dimension} $d$ of a polyhedral complex $K$ is the largest dimension of any cell in $K$.
  In this case, we refer to $K$ as a $d$-complex.
\end{defn}

We will work with \emph{finite} polyhedral complexes, i.e., when the set of cells in $K$ is finite.
While some of our definitions and results apply in arbitrary dimensions, we concentrate mostly on \emph{full-dimensional} polyhedral complexes of dimensions $2$ and $3$, i.e., in $\R^2$ and $\R^3$, respectively.
We will follow the convention that cells up to dimension $3$ are referred to as polyhedra (higher dimensional versions are termed \emph{polytopes}).
Formally, a $d$-dimensional polyhedron (a $d$-cell) is homeomorphic to the closed $d$-dimensional Euclidean ball.
The $d$-cells of interest in this work are vertices ($d=0$), edges ($d=1$), polygons ($d=2$), and \emph{polyhedra} or $3$-cells ($d=3$).

Our definition of Euler transformation (in \cref{sec:eulertsfm}) as well as geometric realization results in $d=2,3$ (in \cref{sec:geomrlzn}) do not require the polyhedra in $K$ to be convex.
Note that vertices and edges are always convex, but polygons and $3$-cells could be nonconvex in our general setting.
Further, some cells in the Euler transformed complex $\hK$ may not to be convex.
But if we assume cells in $K$ are convex, then we can guarantee a large majority of cells in $\hK$ are so as well.
We assume cells in $K$ are convex when describing results on the geometric quality of cells in $\hK$ (in \cref{sec:geomqual}).

\begin{defn} \label{def:purecplx}
  \emph{\bfseries (Pure  complex)}
  A polyhedral $d$-complex is \emph{pure} if every $p$-cell in $K$ for $p < d$ is a face of some $d$-cell in $K$.
\end{defn}
Being pure means that all top-dimensional cells in $K$ have dimension $d$.
A pure $2$-complex has no ``isolated'' edges or vertices, for instance.
In other words, every edge is a face of some polygon in the complex.

We assume the input mesh $K$ is a finite, connected, pure $d$-complex in $\R^d$ for $d=2$ or $d=3$.
Along with $K$, we assume we are given a collection $\CH$ of $d$-cells that capture $d$-dimensional \emph{holes}, and a singleton set $\CO$ that contains a $d$-cell capturing the \emph{outside}.
To be precise, $\CH = \cup_i c_i$ where each $c_i$ is a $d$-cell that is \emph{not} part of $K$ but all $(d-1)$-cells that constitute its $d$-boundary, which is homeomorphic to a $(d-1)$-sphere, are present in $K$.
Note that $p$-cells for $p < d$ in the intersection of a $d$-cell in $K$ and a $d$-cell in $\CH$ or $\CO$ are precisely the boundary cells of $K$.
For technical reasons that we explain later, we make the following assumptions about intersections of full-dimensional cells in $K$, $\CH$, and $\CO$.
We denote the \emph{underlying spaces} of these objects as $|K|, |\CH|$, and $|\CO|$, respectively.
Note that $|\CH| = \cup_{c_i \in \CH} |c_i|$.

\begin{asmn} \label{asmn:Kholesoutside}
  The following conditions hold for the input complex $K$, the collection of holes $\CH$, and the outside cell $\CO$.
  \begin{enumerate}
    \item $|K| \cup |\CH| \cup |\CO| = \R^d$.
    \item \label{asmn:sepholes} $d$-cells in $\CH$ are pairwise disjoint, and are also disjoint from the $d$-cell that is $\CO$. 
    \item \label{asmn:holeintr} Any $d$-cell in $K$ and a $d$-cell in $\CH$ intersect in at most \emph{one} $(d-1)$-facet of both.
      See \cref{rem:adjbdyedges} for an explanation of the need for this assumption.
    \item  \label{asmn:outintr} No two $(d-1)$-cells that are \emph{adjacent facets} of a $d$-cell in $K$, i.e., they intersect in a common $(d-2)$ cell, intersect the $d$-cell that is $\CO$.
      Again, see \cref{rem:adjbdyedges} for an explanation.
  \end{enumerate}
\end{asmn}
\noindent Intuitively, the $d$-cells in $K, \CH$, and $\CO$ cover all of $\R^d$, and each $d$-cell in $\CH$ captures a separate hole that is also separate from the outside.

We point out that \emph{articulation} (or cut) vertices are allowed in $K$, i.e., vertices whose removal disconnects the complex (we assume $K$ is connected to start with).
Conditions specified in \cref{asmn:Kholesoutside} ensure such vertices are boundary vertices of $K$.
For instance, $K$ could consist of two copies of the complex shown on the left in \cref{fig:eulermeshillust} that meet at one of the four corner points.

\section{Definition of Euler Transformation} \label{sec:eulertsfm}

We define the Euler transformation $\hK$ of the input $d$-complex $K$ by explicitly listing the $d$-cells that are included in $\hK$.
Since we are working with cells (rather than simplices), we specify each $d$-cell by explicitly listing all $(d-1)$-cells that are its facets.
We denote vertices as $v$ (or $u, v_i$), edges as $e$ (or $e_i$), polygons or $2$-cells as $f$ (or $f_i$), and $3$-cells as $t$ (or $t_i$).
The corresponding cells in $\hK$ are denoted $\hv, \he, \hf, \htt$, and so on.
We first define the cells in $\hK$ \emph{abstractly}, and discuss aspects of geometric realization in \cref{sec:geomrlzn}.

\vspace*{-0.05in}
\subsection{Euler transformation for $d=2$} \label{ssec:euler2d}
\vspace*{-0.05in}

We start by \emph{duplicating} every polygon ($2$-cell) in $K \cup \CH \cup \CO$.
Since we do not want to alter the domain in $\R^d$ captured by $K$, we set $\hCH = \CH$ and $\hCO = \CO$.
But we ``shrink'' each polygon in $K$ when duplicating (see \cref{sec:geomrlzn} for details).
By the definition of $K$ and \cref{asmn:Kholesoutside}, this duplication results in each edge $e \in K$ being represented by two copies in $\hK$.

The polygons ($2$-cells) in $\hK$ belong to three classes, and correspond to the polygons, edges, and vertices in $K$ as described below.
See Figure \ref{fig:3types2cells2d} for illustrations of each class.
\begin{enumerate}
  \item \label{2dETcls1}
    For each polygon $f \in K$, we include $\hf \in \hK$ as the copy of $f$.

  \item \label{2dETcls2}
    Each edge $e \in K$ generates the $4$-gon ($4$-sided polygon) $\hf_e$ in $\hK$ specified as follows.
    Let $e = \{u,v\} \in f,f'$, where $f \in K$ and $f' \in K \cup \CH \cup \CO$.
    Then $\hf_e$ is the polygon whose facets are the four edges $\{\hu,\hv\}, \, \{\hv,\hv'\}, \, \{\hu',\hv'\}$, and $\{\hu,\hu'\}$.
    Here, $\hv, \hv'$ are the two copies of $v$ in $\hK$.
    Note that the edges $\he = \{\hu,\hv\}$ and $\he' =  \{\hu',\hv'\}$ are facets of the Class \ref{2dETcls1} polygons $\hf$ added to $\hK$ (as described above) or of the polygons $\hf'$ in $\hCH$ or $\hCO$.
    Edges $\{\hu,\hu'\}$ and $\{\hv,\hv'\}$ are added new.
    
  \item \label{2dETcls3}
    Each vertex $v \in K$ that is part of $p$ polygons in $K$ generates a $p$-gon (polygon with $p$ sides) $\hf_v$ in $\hK$ whose vertices and edges are specified as follows.
    Let $v \in f_k$ for $k=1,\dots,p$ in $K$.
    Then $\hf_v$ has vertices $\hv_k$, $k=1,\dots,p$, where $\hv_k$ is the copy of $v$ in $\hf_k$ (in $\hK$).
    For every pair of polygons $f_i,f_j \in \{f_k\}_1^p$ that intersect in an edge $e_{ij} \in K$, the edge $\he_{ij} = \{\hv_i, \hv_j\}$ is included as a facet of $\hf_v$.
    Note that edges $\he_{ij}$ are precisely the edges added new as facets of the Class \ref{2dETcls2} polygons described above.
\end{enumerate}
\begin{figure}[htp!]
  \centering
  \includegraphics[scale=0.2]{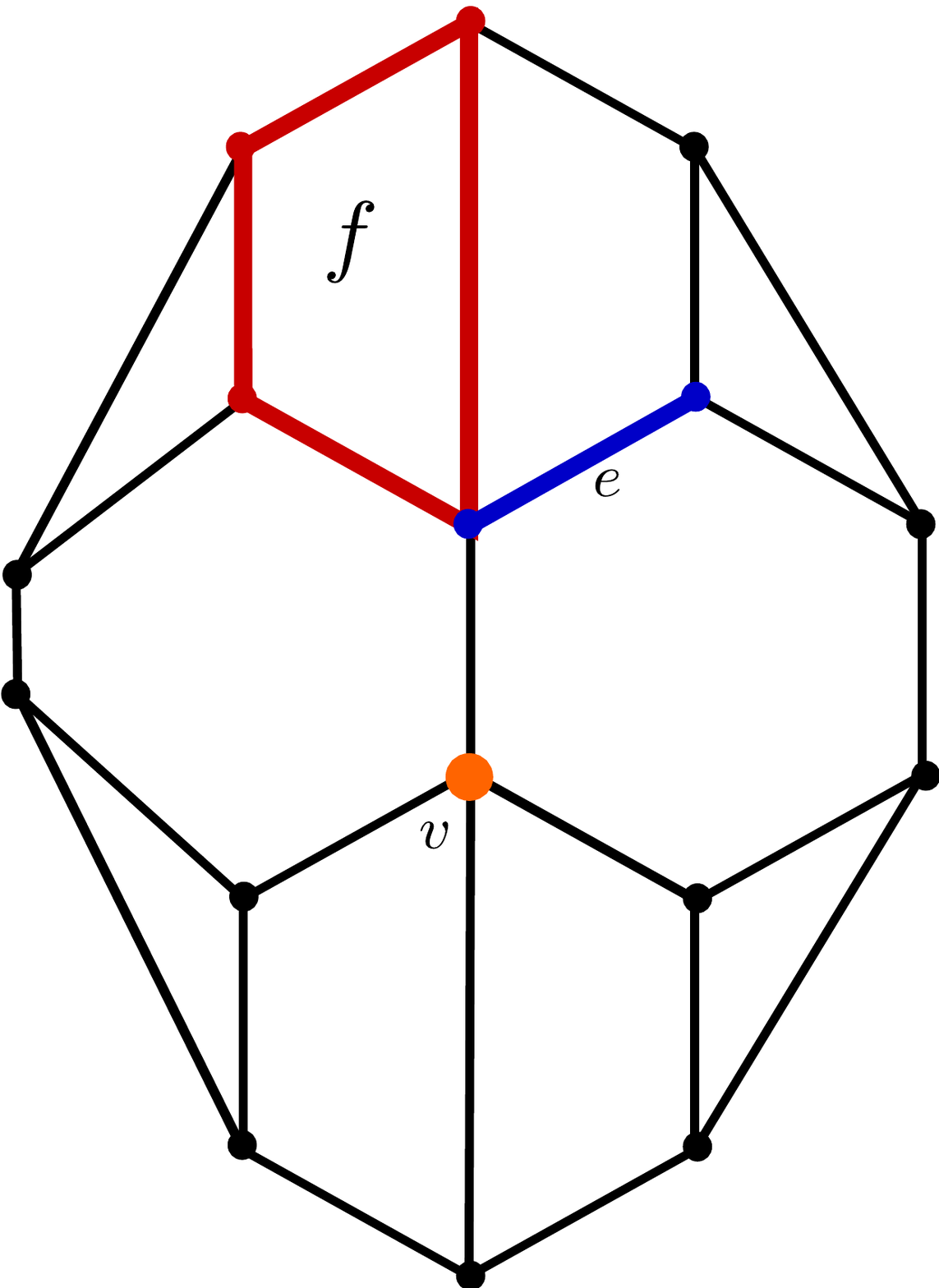}
  \quad
  \includegraphics[scale=0.2]{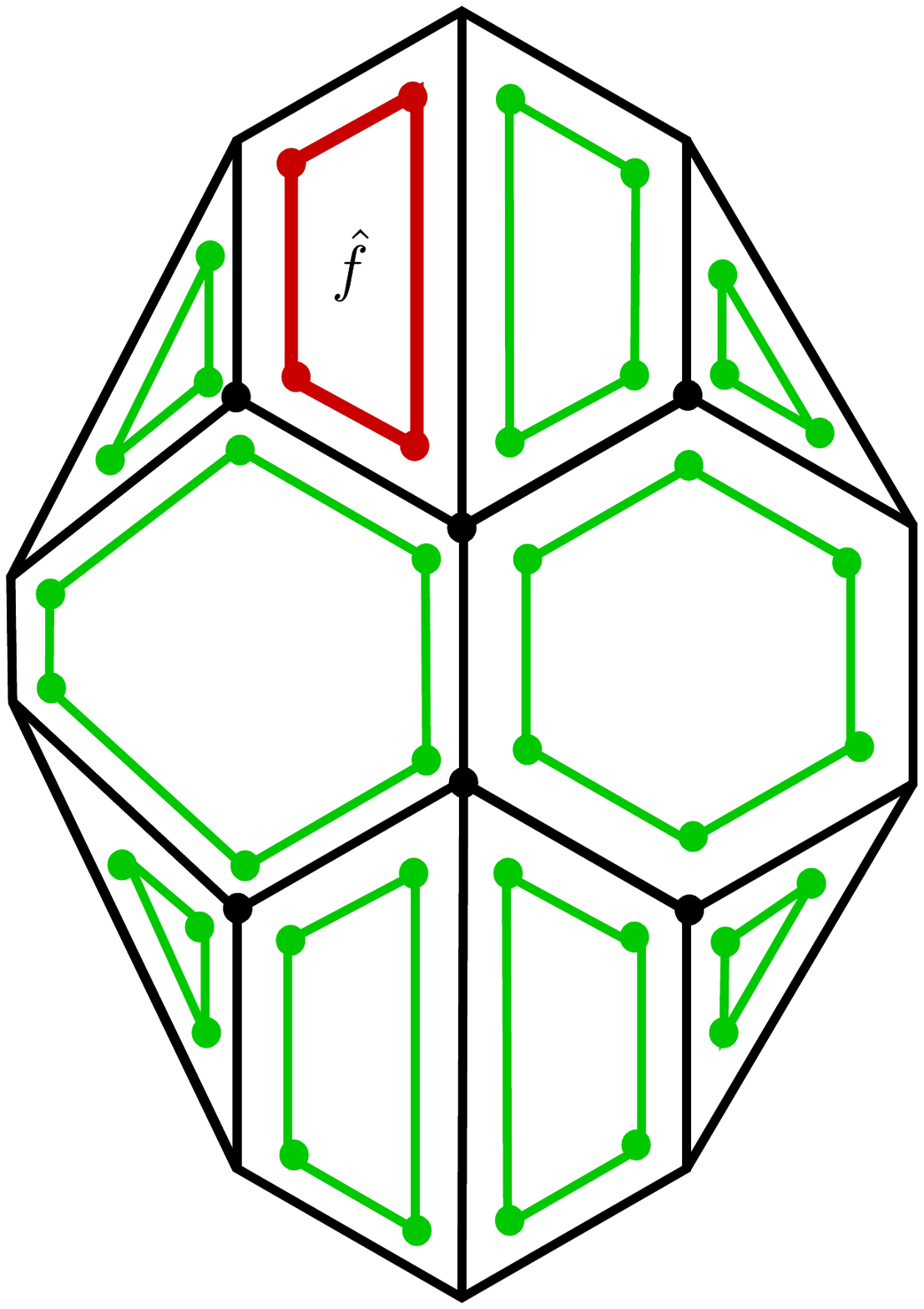}
  \quad
  \includegraphics[scale=0.2]{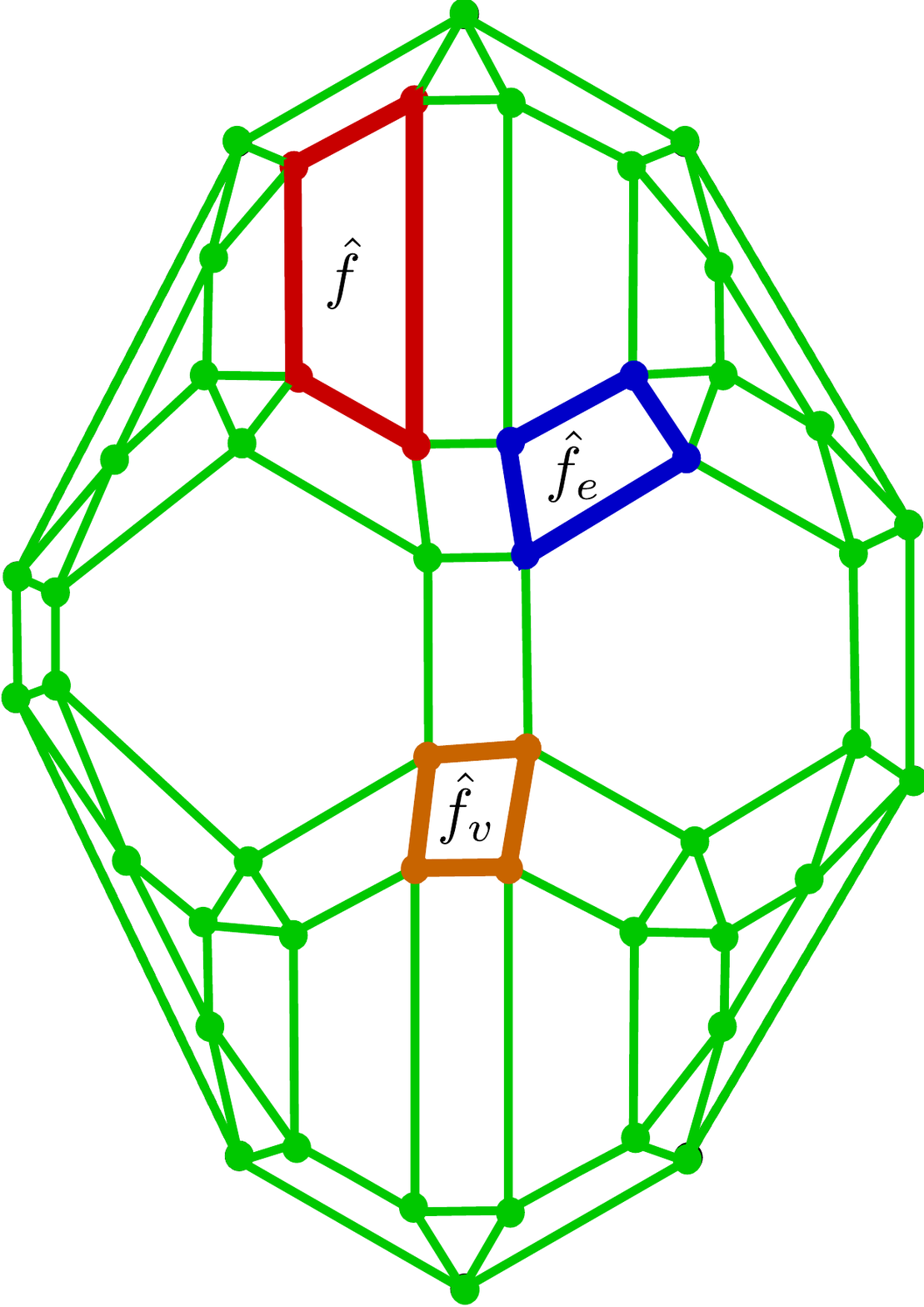}
  \caption{A polygon $f$, edge $e$, and a vertex $v$ highlighted in an input complex $K$ (left),
    an intermediate complex showing only the copies of original polygons in $K$ that are included in $\hK$, i.e., of Class \ref{2dETcls1} (middle),
    and the final Euler transformation $\hK$ (right).}
  \label{fig:3types2cells2d}
\end{figure}

\subsection{Euler transformation for $d=3$} \label{ssec:euler3d}

We start by duplicating every polyhedron ($3$-cell) in $K \cup \CH \cup \CO$.
Similar to the case of $d=2$, we set $\hCH = \CH$ and $\hCO = \CO$, but ``shrink'' each $3$-cell in $K$ when duplicating.
This duplication results in each polygon $f \in K$ being represented by two copies in $\hK$.

The $3$-cells in $\hK$ belong to four classes, and correspond to the $3$-cells, polygons, edges, and vertices in $K$ as described below.
%
%\clearpage
\begin{enumerate}
  \item \label{3dETcls1}
    For each $3$-cell $t \in K$, we include $\htt \in \hK$ as the copy of $t$.
    
  \item \label{3dETcls2}
      Each $p$-sided polygon $f \in K$ generates a $3$-cell $\htt_f$ in $\hK$ with $p+2$ polygons as facets specified as follows.
      Let $f$ be a facet of $3$-cells $t, t'$ where $t \in K$ and $t' \in K \cup \CH \cup \CO$.
      Also, let $f$ consist of vertices $\{v_i\}_{i=1}^p$ and edges $e_i=\{v_i,v_j\}$ for $i=1,\dots,p-1,\,j=i+1$ and $i=p,j=1$.
      Then the polygons that are facets of $\htt_f$ include $\hf, \hf'$, and the $p$ $4$-gons $\hf_i$ whose edges consist of $\he_i,\he'_i, \{\hv_i,\hv_i'\}, \{\hv_j, \hv_j'\}$ for $i=1,\dots,p-1,\,j=i+1$ and $i=p,j=1$.
      Here, $\hv'_i, \he'_i$ are the vertex and edge in $\hf'$ generated by $v_i, e_i$ for each $i$, where $\hf'$ is the polygonal facet of $\htt'$ corresponding to $f$.
      And $\htt'$ is the Class \ref{3dETcls1} $3$-cell in $\hK$ (as defined above) corresponding to $t'$.

      We also note that the polygons $\hf,\hf'$ are already included as facets of $\htt, \htt'$, which are added to $\hK$ as Class \ref{3dETcls1} cells.
      The vertices $\hv_i, \hv'_i$ are included already as part of $\htt, \htt'$ as well.
      The $p$ $4$-gons $\hf_i$ are added \emph{new}, and so are the edges $\{\hv_i,\hv'_i\}$ for $i=1,\dots,p$.
      See Figure \ref{fig:3dETcls2} for an illustration with $p=5$.
      \begin{figure}[htp!]
        \centering
        %\vspace*{-0.1in}
        \includegraphics[scale=0.2]{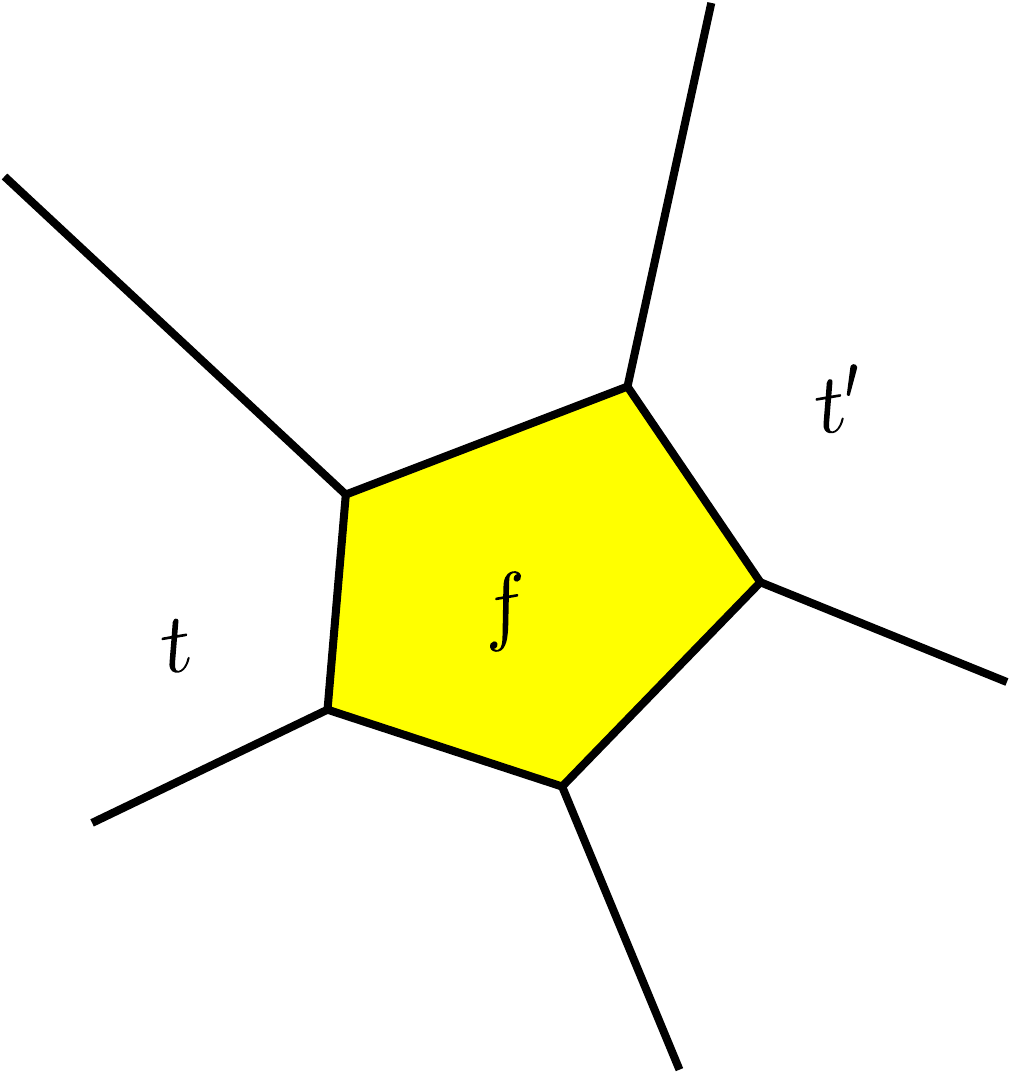}
        \hspace*{0.3in}
        \includegraphics[scale=0.2]{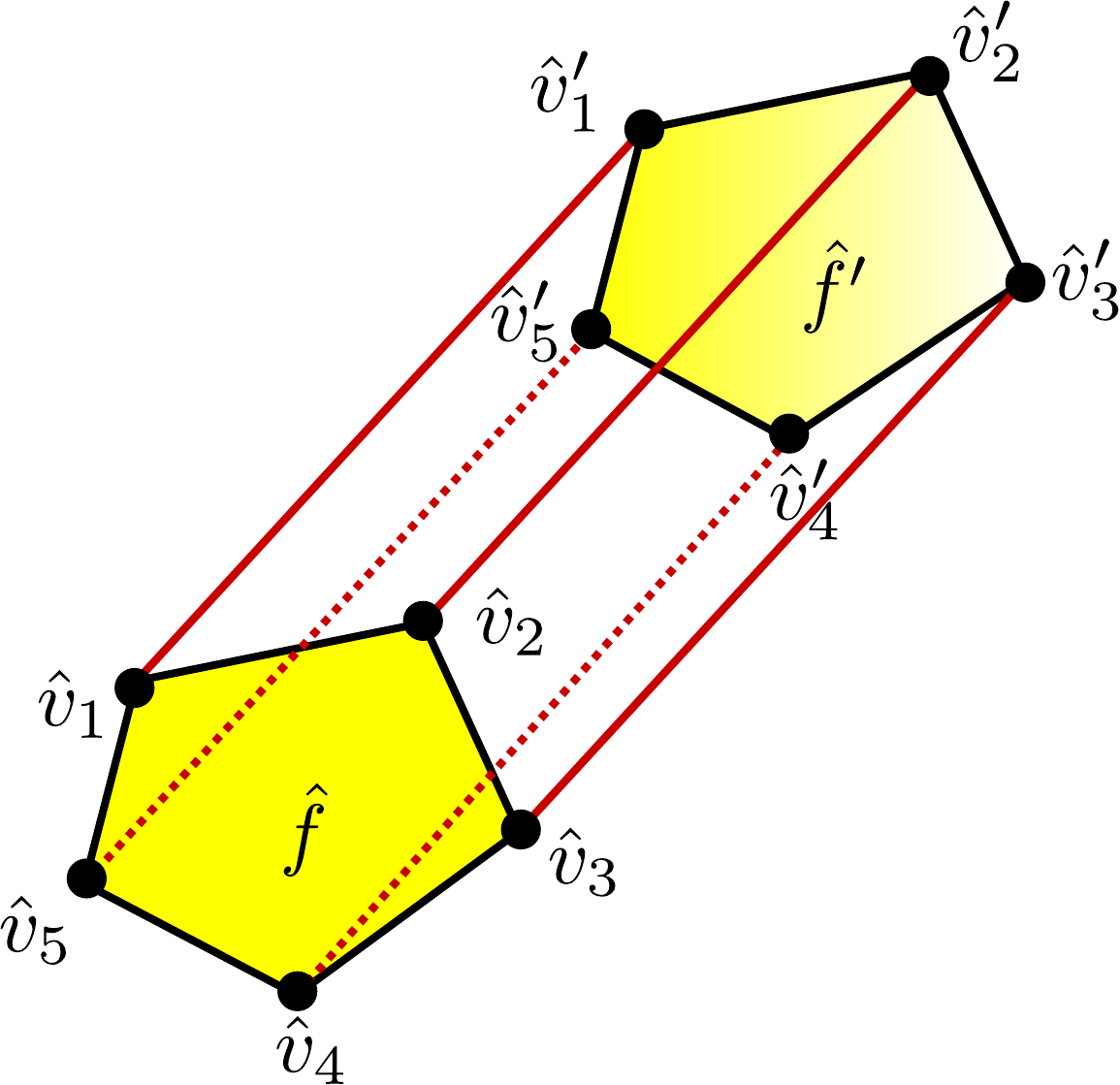}
        \caption{Illustration of a Class \ref{3dETcls2} $3$-cell $\htt_f$ (right) generated by the pentagon $f \in K$ (left).
          The $p=5$ $4$-gons are not shaded for clarity.}
        \label{fig:3dETcls2}
      \end{figure}
      %\vspace*{-0.1in}
      
  \item \label{3dETcls3}
    Each edge $e \in K$ that is part of $q$ $3$-cells generates a $3$-cell $\htt_e$ in $\hK$ with $q+2$ polygons as facets specified as follows.
    Let $\{t_i\}_{i=1}^q$ be the $q$ $3$-cells that have $e = \{u,v\}$ as an edge, and
    let $\he_i = \{\hu_i, \hv_i\},\,i=1,\dots,q$ be the $q$ copies of the edge $e$ in $\htt_i$.
    We add a $4$-gon $\hf_i$ as a facet of $\htt_e$ for every pair of \emph{adjacent} $3$-cells $t_i,t_j$ from this collection, i.e., when $t_i \cap t_j = f_{ij}$ is a polygon that contains $e$ as an edge, for $i=1,\dots,q-1,\,j=i+1$ and $i=q,j=1$.
    The edges of this $4$-gon $\hf_i$ are $\he_i, \he_j, \{\hu_i, \hu_j\}$, and $\{\hv_i, \hv_j\}$.
    Note that there are $q$ such $4$-gons $\hf_i$ that are facets of $\htt_e$.
    Also note that these $4$-gons are already included as faces of the Class \ref{3dETcls2} $3$-cells $\htt_{f_{ij}}$ described above.

    Finally, we add two \emph{new} $q$-gons as facets of $\htt_e$ whose $q$ edges are $\{\hu_i,\hu_j\}$ and $\{\hv_i,\hv_j\}$, respectively, for  $i=1,\dots,q-1,\,j=i+1$ and $i=q,j=1$.
    \cref{fig:3dETcls3} illustrates a $q=5$ instance.
    At the same time, these $q$-gons cannot be guaranteed to be planar in all geometric realizations of $\hK$ (see \cref{rem:nonplnr3d}). % in \cref{ssec:rlzn3d}).
    \begin{figure}[htp!]
      \centering
      \includegraphics[scale=0.2]{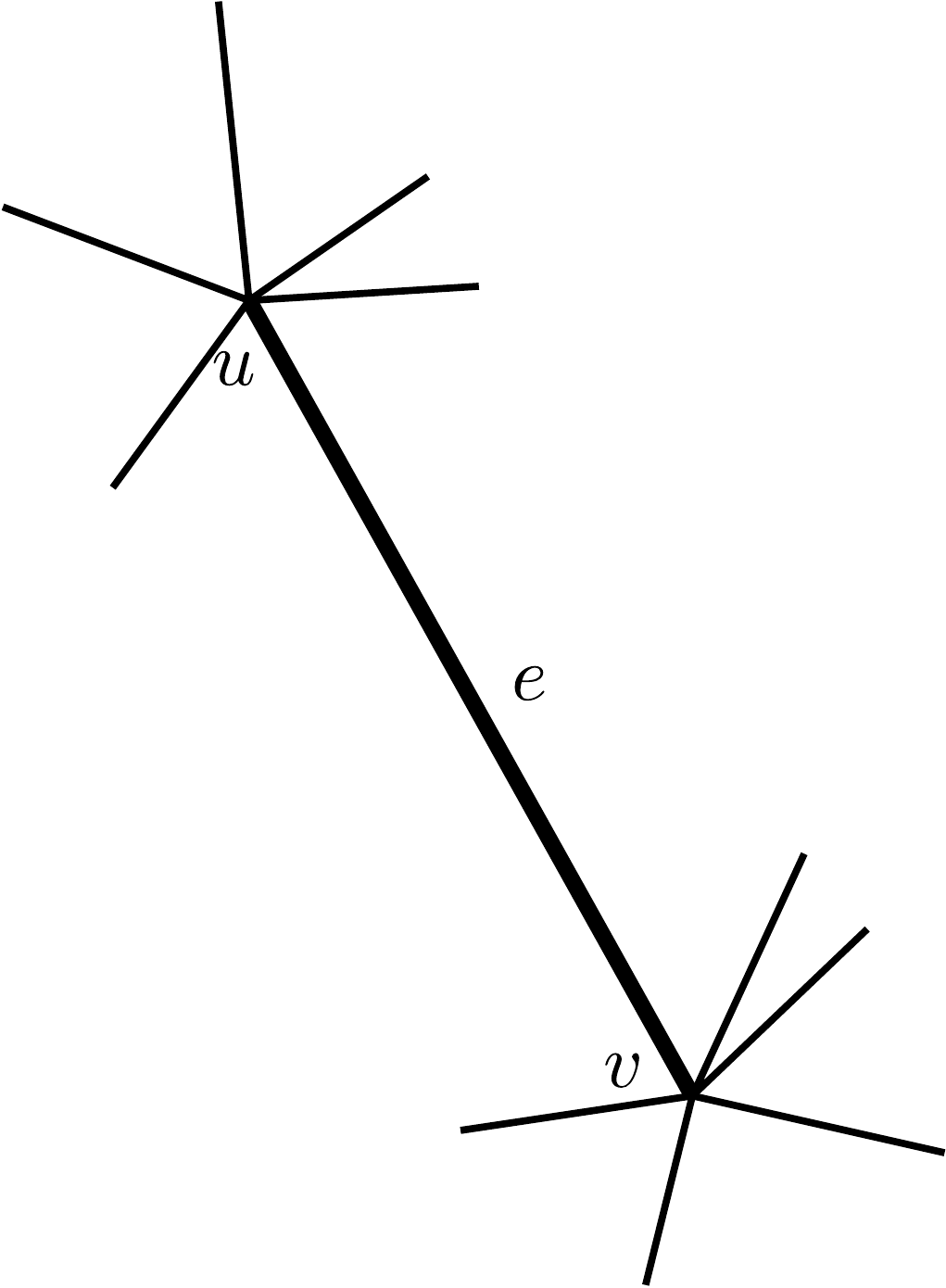}
      \hspace*{0.1in}
      \includegraphics[scale=0.2]{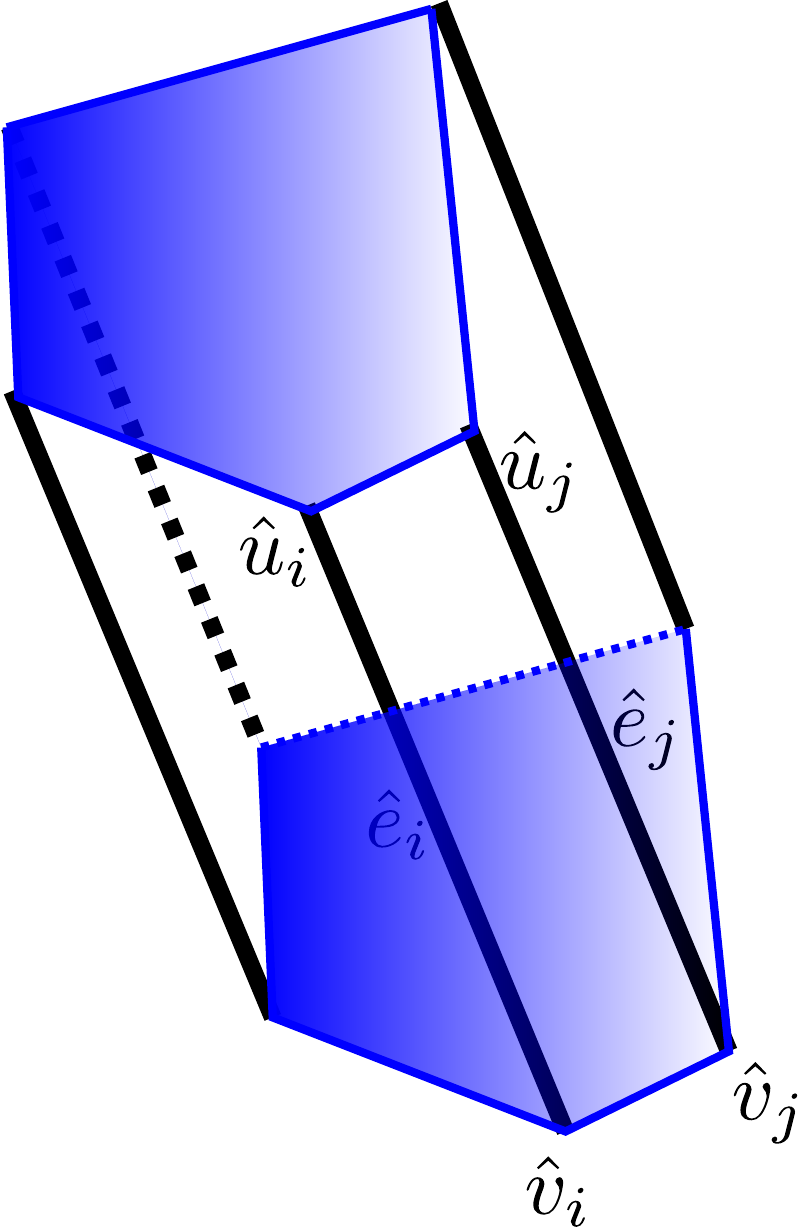}
      \caption{Illustration of a Class \ref{3dETcls3} $3$-cell $\htt_e$ (right) generated by the edge $e \in K$ shared by $5$ cells (left).
        The $q=5$ $4$-gons are not shaded for clarity.
      }
        \label{fig:3dETcls3}
    \end{figure}

  \item \label{3dETcls4}
    Each vertex $v \in K$ generates a $3$-cell $\htt_v$ in $\hK$ specified as follows.
    Let $v$ be a vertex in $r$ $3$-cells $\{t_k\}_{k=1}^r$, and let $\hv_k$ be the corresponding copy of $v$ in $\htt_k$ (in $\hK$) for $k=1,\dots,r$.
    Then the vertices of $\htt_v$ are precisely $\hv_1,\dots,\hv_r$.
    For every pair of \emph{adjacent} $3$-cells $t_i,t_j \in \{t_k\}_{k=1}^r$, i.e., when $t_i \cap t_j$ is a polygon that contains $v$ as a vertex, we add the edge $\{\hv_i,\hv_j\}$ to $\htt_v$.

    We then consider every subset $\sT \subset \{t_k\}_{k=1}^r$ of $3$-cells that intersect in an edge that has $v$ as one vertex.
    We add a polygon as a facet of $\htt_v$ that has as edges $\{\hv_i,\hv_j\}$ where $t_i,t_j \in \sT$ intersect in a polygon.
    We repeat this process for every such $\sT \subset \{t_k\}_{k=1}^r$ with $3 \leq |\sT| < r$.
    Note that each such polygon has already been added as a facet of the Class \ref{3dETcls3} $3$-cell generated by the edge common to the $3$-cells in collection $\sT$ (these are the $q$-gons described above).

\end{enumerate}

\begin{rem}
  Based on our definition of $\hK$, any two cells in $\hK$ intersect in a face of both cells, and all faces of cells are included in $\hK$.
  Hence $\hK$ is a polyhedral complex.
\end{rem}

\begin{rem}
 \cref{thm:deg4ET2d}, \cref{lem:cntshVhEhF2d}, \cref{thm:deg6ET3d}, and \cref{lem:cntshThVhEhf3d} hold independent of the geometric realization of the Euler transformed complex.
\end{rem}

\section{Geometric Properties of the Euler Transformed Complex} \label{sec:geomrlzn}

As Euler transformation adds new full-dimensional cells corresponding to polygonal facets, edges, and vertices, we \emph{offset} the cells added to $\hK$ as copies of the full-dimensional cells in $K$ in order to generate enough space to add the extra cells.
Intuitively, we ``shrink'' each of the full-dimensional cells in $K$ in order to produce cells in $\hK$ that are geometrically similar to the input cells.
We use standard techniques for producing offset polygons in 2D, e.g., \emph{mitered offset} generated using the straight skeleton (SK) of the input polygon \cite{AiAuAlGa1995}.

%\vspace*{-0.2in}
\begin{defn} \label{def:miteredoffset}
  \emph{\bfseries (Mitered offset \cite{AiAuAlGa1995})}
  The mitered offset of a polyhedron $P$ with offset distance $b$ is the polyhedron obtained by moving each facet of $P$ toward its interior in the direction of the normal to the supporting plane of the facet and where the resulting facet has a supporting plane at a distance $b$ along the normal. 
\end{defn}

For the $d=2$ case, we define the cell offset as a mitered offset of the polygon (i.e., choice of $b$ in \cref{def:miteredoffset}) that creates no \emph{combinatorial} or \emph{topological changes}---i.e., no edges are shrunk to points, and the polygon is not split into multiple polygons.
On the other hand, we have studied the case of mitered offset with \emph{combinatorial} or \emph{topological changes} in $d=2$ separately \cite{GuKrDr2020}.

Unlike for polygons, parallel offsetting of a polyhedron ($d=3$) is not defined uniquely in general \cite{BaEpGoVa2008}.
Although a unique offset polyhedron could be constructed for orthogonal polyhedra or for convex polyhedra \cite{BaEpGoVa2008,MaViPl2011}, shrinking a generic polyhedron goes through continuous geometrical changes, combinatorial changes, as well as topological changes (e.g., breaking into multiple polyhedra).
In fact, offsetting vertices with degree $4$ or more in the $1$-skeleton of the polyhedron can produce multiple vertices even with an infinitesimal shrinkage \cite{AuWa2013,AuWa2016}.
Hence we assume in the case of $d=3$ that each vertex in the input cell complex $K$ has degree $3$ in the $1$-skeleton of each cell that it is part of, on top of the requirements in \cref{asmn:Kholesoutside}
(see \cref{rem:3dETdeg4vtx} for how one may deal with a complex where this assumption does not hold).
Under this assumption, we define the cell offset of a polyhedron as a mitered offset that creates no combinatorial or topological changes (similar to the case of polygons in $d=2$).

Naturally, we do not want to alter the domain modeled by the cell complex $K$, i.e., its underlying space $|K|$.
Hence we maintain the cells in $\CH$ and $\CO$, i.e., these cells are included in $\hK$ without any changes.
Since every top-dimensional cell is offset, the new cells in $\hK$ are \emph{fit} within the extra space created by offsetting each top-dimensional cell.

\subsection{Geometric Realization in $d=2$} \label{ssec:rlzn2d}

We state and prove several properties of the geometric realization of the Euler transformed polygon complex $\hK$ in $d=2$.
We restrict our discussion to the cases where $2$-cells in $\hK$ are planar and edges are straight lines.
We start with the main result---every vertex in $\hK$ has degree $4$ in its $1$-skeleton.

\begin{thm} \label{thm:deg4ET2d}
  Every vertex in $\hK$, the Euler transformation of the $2$-complex $K$, has degree $4$ in the $1$-skeleton of $\hK$.
\end{thm}
\begin{proof}
  Consider a vertex $v$ shared by adjacent edges $e_1, e_2 \in f$, where $f \in K \cup \CH \cup \CO$ is a polygon.
  Following \cref{asmn:Kholesoutside}, the edges $e_1$ and $e_2$ are shared by exactly two cells each from the input complex,  holes, or the outside cell.
  Let $f'_1, f'_2$ be the other polygons containing edges $e_1, e_2$, respectively (with $f$ being the first polygon).

  Consider the vertex $\hv \in \hK$ generated as part of $\hf$, as specified in the Euler transformation (\cref{ssec:euler2d}).
  $\hf$ is a mitered offset of $f$ when $f \in K$, or is identical to $f$ when it belongs to $\CH \cup \CO$.
  Hence $\hf$ is a simple polygon in both cases, and $\hv$ is part of two edges $\he_1, \he_2 \in \hf$.
  Further, $\hv$ will be part of two more edges $\{\hv,\hv'_1\}$ and  $\{\hv,\hv'_2\}$  added as facets of the Class \ref{2dETcls2} polygons generated by $e_1, e_2$.
  Here $\hv'_i \in \he'_i \in \hf'_i$ for $i=1,2$.
  Hence $\hv$ has degree $4$ in the $1$-skeleton of $\hK$.
\end{proof}

\begin{rem}
  \label{rem:adjbdyedges}
  {\rm 
    We show why we require the input complex to satisfy Conditions \ref{asmn:holeintr} and \ref{asmn:outintr} in \cref{asmn:Kholesoutside}, which require that no two adjacent edges of a polygon in $K$ can be boundary edges.
    Consider the input complex $K$ consisting of a single square, whose four edges are shared with the outside cell $\CO$.
    If we apply the Euler transformation as specified in \cref{ssec:euler2d}, every vertex in the output complex  $\tilde{K}$ will have the odd degree of $3$, as shown in Figure \ref{fig:outintr}.
    But if we apply the Euler transformation \emph{once more} to $\tilde{K}$, we do get a valid complex $\hK$ with each vertex having degree $4$.
    Note that $\tilde{K}$ does satisfy Condition \ref{asmn:outintr}, and hence becomes a valid input.
    \begin{figure}[htp!]
      \centering
      \includegraphics[scale=0.2]{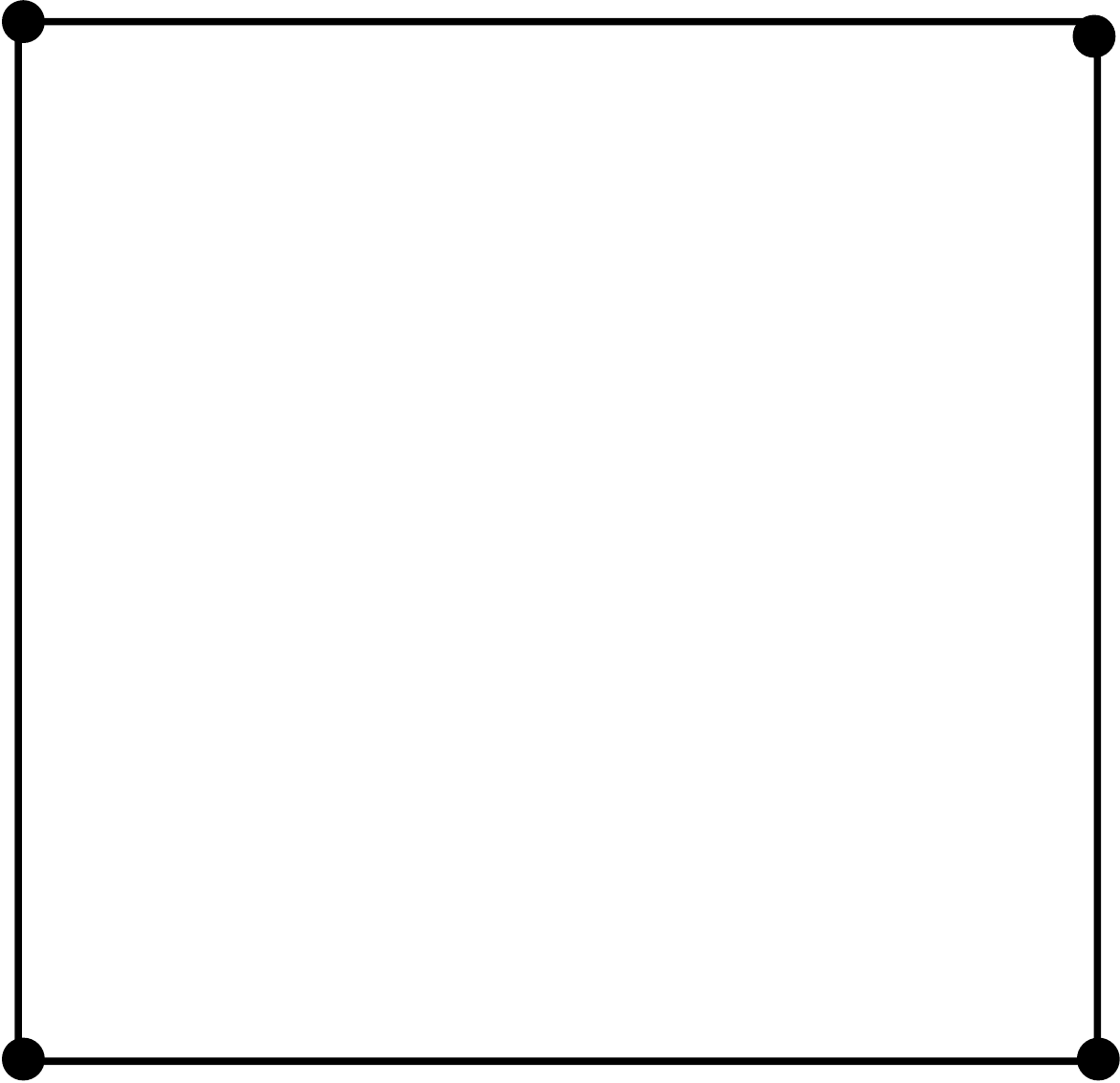}
      \quad
      \includegraphics[scale=0.2]{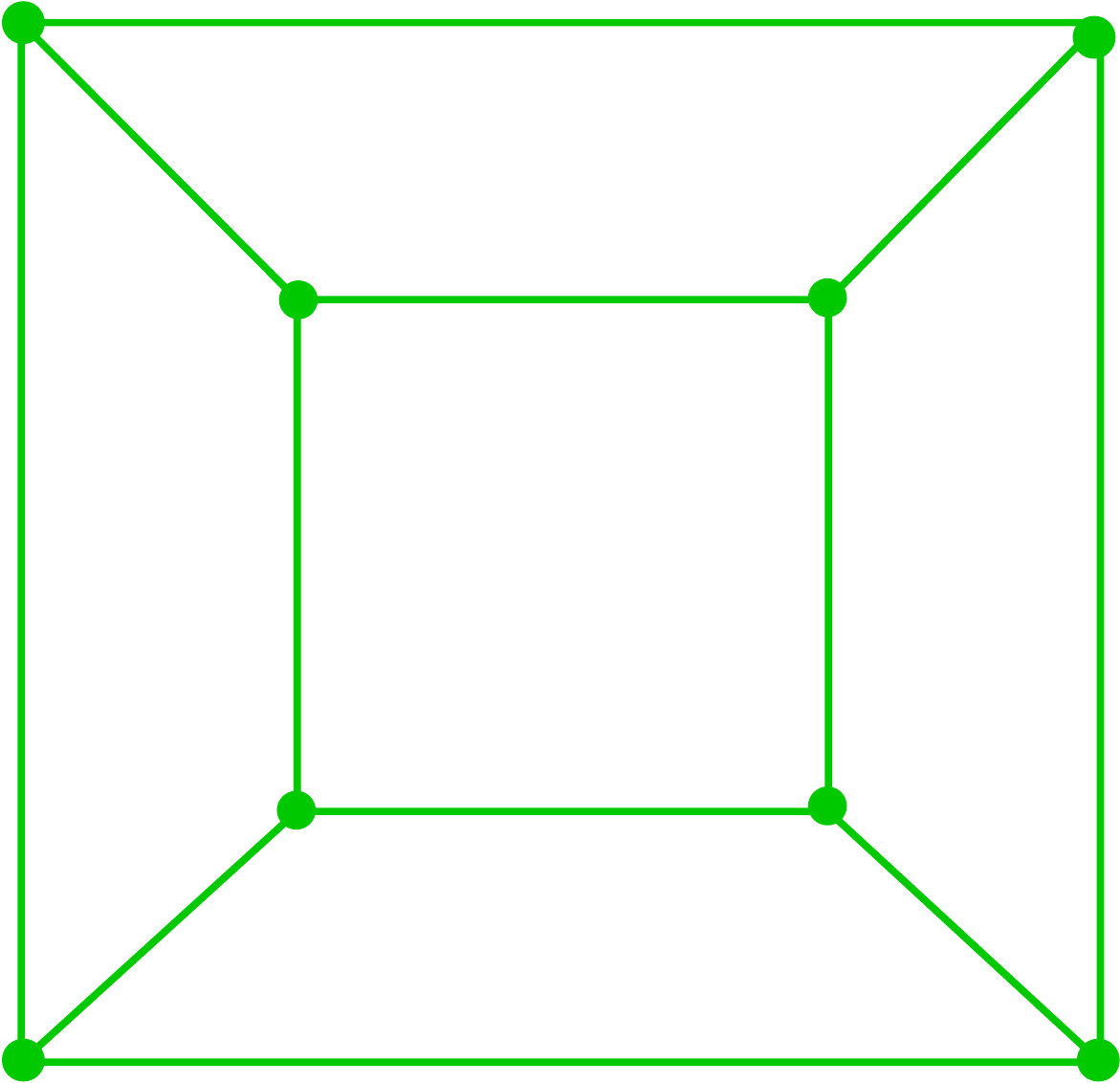}
      \quad
      \includegraphics[scale=0.2]{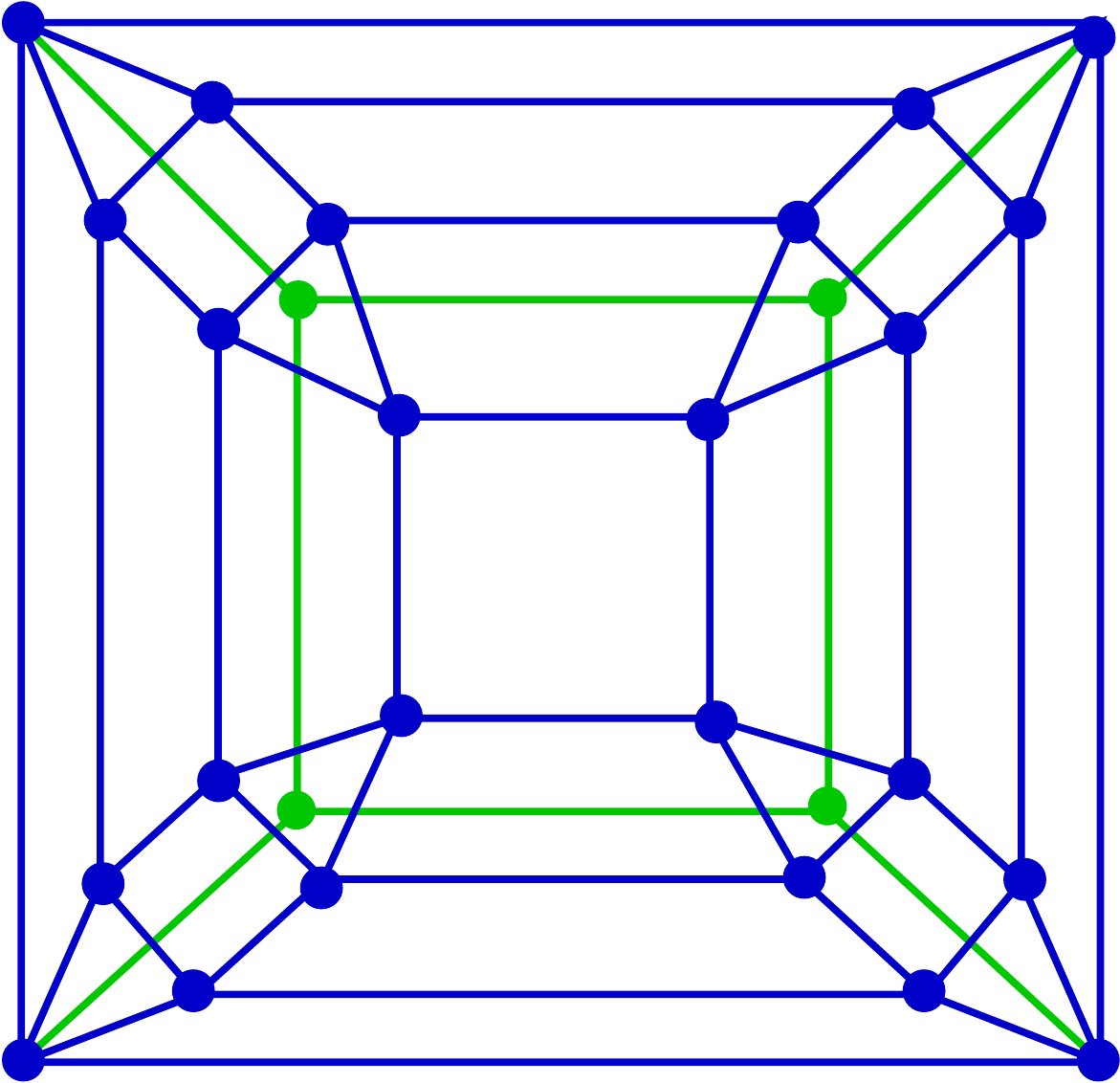}
      \caption{Applying the Euler transformation to the square $K$ whose more than one adjacent edge is shared with the outside (left) produces a complex $\tilde{K}$ in which every vertex has odd degree (middle).
      Applying the Euler transformation again to $\tilde{K}$ produces a valid complex $\hK$ where every vertex has degree $4$ (right).}
      \label{fig:outintr}
    \end{figure}
  }
\end{rem}

\begin{lem}
  \label{lem:cntshVhEhF2d}
  Let $V, E, F$ denote the sets of vertices, edges, and polygons (faces) in $K$, and let $\hV, \hE, \hF$ denote the corresponding sets in $\hK$.
  The following relations hold for the cardinalities of these sets: $|\hV| = 2|E|,~|\hE| = 4|E|$, and $|\hF| = |V| + |E| + |F|$.
\end{lem}

\begin{proof}
  Let $\delta(v)$ denote the degree of vertex $v \in K$, and let $\hf_v$ be the polygon generated by $v$ in $\hK$.
  This is a polygon of Class \ref{2dETcls3} specified in the definition of Euler transformation for $d=2$ (Section \ref{ssec:euler2d}).
  Following \cref{asmn:Kholesoutside} about $K, \CH, \CO$, it is clear that when $v$ belongs to $p$ polygons in $K$, we must have $\delta(v)=p$ and $\hf_v$ has $p$ vertices.
  Since each cell $\hf$ corresponding to polygon $f \in K$ is a mitered offset, and since each vertex $\hv$ is part of one such offset polygon, it follows that $\hf_u \cap \hf_v = \emptyset$ for any two vertices $u,v \in K$.
  Hence we get
  \[
    |\hV| = \sum_{v \in K} \delta(v) = 2|E|.
  \]
  By \cref{thm:deg4ET2d}, each vertex $\hv \in \hK$ has degree $\hat{\delta}(\hv) = 4$ in $\hK$.
  Combined with the result above on $|\hV|$, we get that
  \[
    |\hE| = \frac{1}{2} \sum_{\hv \in \hK} \hat{\delta}(\hv) = \frac{1}{2} \cdot 2|E| \cdot 4 = 4|E|.
  \]
  Following the definition of Euler transformation (Section \ref{ssec:euler2d}), each polygon, edge, and vertex in $K$ generate corresponding unique polygons in $\hK$ belonging to three classes.
  Hence we get $|\hF| = |F| + |E| + |V|$.
\end{proof}

\begin{lem}
  \label{lem:hGplanar}
  Let $\hG$ denote the graph that is the $1$-skeleton of $\hK$. Then $\hG$ is planar.
\end{lem}
\begin{proof}
  By the definition of Euler transformation (\cref{ssec:euler2d}), and since each polygon $\hf \in \hK$ generated by the polygon $f \in K$ is a mitered offset and hence a simple closed polygon, any two polygons $\hf, \hf' \in \hK$ of Class \ref{2dETcls1} generated by polygons $f,f' \in K$ satisfy $\hf \cap \hf' = \emptyset$.

  Consider two polygons $\hf_e, \hf_{e'} \in \hK$ of Class \ref{2dETcls2} generated by edges $e,e' \in K$.
  By the way we construct these polygons, $\hf_e$ and $\hf_{e'}$ intersect at a vertex $\hv$ if and only if $e$ and $e'$ are adjacent edges of a polygon $f \in K$ meeting at the vertex $v$.

  Since each $\hf \in \hK$ is a mitered offset of some polygon $f \in K$, at least one of the two copies $\he, \he'$ of edges in $\hK$ corresponding to the edge $e \in K$ is shorter in length than $e$ (see Definition of Class \ref{2dETcls2} polygons).
  In particular, if $e$ is not a boundary edge, then both $\he$ and $\he'$ are shorter than $e$.
  If $e$ is a boundary edge, i.e., $e \in f \in \CH \cup \CO$, then one edge out of $\he, \he'$ has the same length as $e$ while the other is shorter.
  Hence each polygon $\hf_e$ of Class \ref{2dETcls2} is a convex $4$-gon (trapezium).

  Since all edges of the polygon $\hf_v$ of Class \ref{2dETcls3} generated by vertex $v \in K$ are precisely the \emph{new} edges added to define the Class \ref{2dETcls2} polygons, each $\hf_v$ is a simple closed polygon.
  Further, by the properties of Class \ref{2dETcls2} polygons specified above, $\hf_v \cap \hf_{v'} = \emptyset$ for any two vertices $v, v' \in K$.

  Thus every polygon in $\hK$ is simple and closed.
  Any two such polygons intersect at most in an edge or a vertex, and any two edges in $\hK$ intersect at most in a vertex.
  Hence $\hG$, the $1$-skeleton of $\hK$, is a planar graph.
\end{proof}

We do not alter the holes or the outside cell.
Hence $\hCH = \CH$ and $\hCO = \CO$, by definition, and $|\hK| \cup |\hCH| \cup |\hCO| = \R^2$ as expected.
Further, using the counts of cells in $\hK$ specified in \cref{lem:cntshVhEhF2d}, we get
\[
  |\hV| - |\hE| + |\hF| = 2|E| - 4|E| +  |V| + |E| + |F| = |V| - |E| + |F|,
\]
confirming that the Euler characteristic remains unchanged by the transformation.
Since the input complex is assumed to be planar, this result reconfirms the planarity of the output complex.

\begin{rem}
  \label{rem:disjholes}
  {\rm
    We illustrate why we require holes in the domain to be disjoint (Condition \ref{asmn:sepholes} in \cref{asmn:Kholesoutside}).
    Consider the input complex $K$ with two holes $h, \bar{h} \in \CH$ that intersect at a vertex $v$.
    The corresponding vertex $\hv$ in the transformed complex $\hK$ will not have a degree of $4$.
    There will also be other vertices in $\hK$ that have odd degree, which are circled in \cref{fig:disjholes}.
    Let these odd-degree vertices be labeled $\hv',\hv''$.
    Technically, there are two \emph{identical} copies of the edge $\{\hv,\hv'\}$ and similarly of $\{\hv,\hv''\}$.
    But such duplicate edges make the graph $\hG$ ($1$-skeleton of $\hK$) non-planar.
    If we include only one copy of each pair of duplicate edges, we get odd degree vertices in $\hG$.    \begin{figure}[htp!]
      \centering
      \includegraphics[scale=0.3]{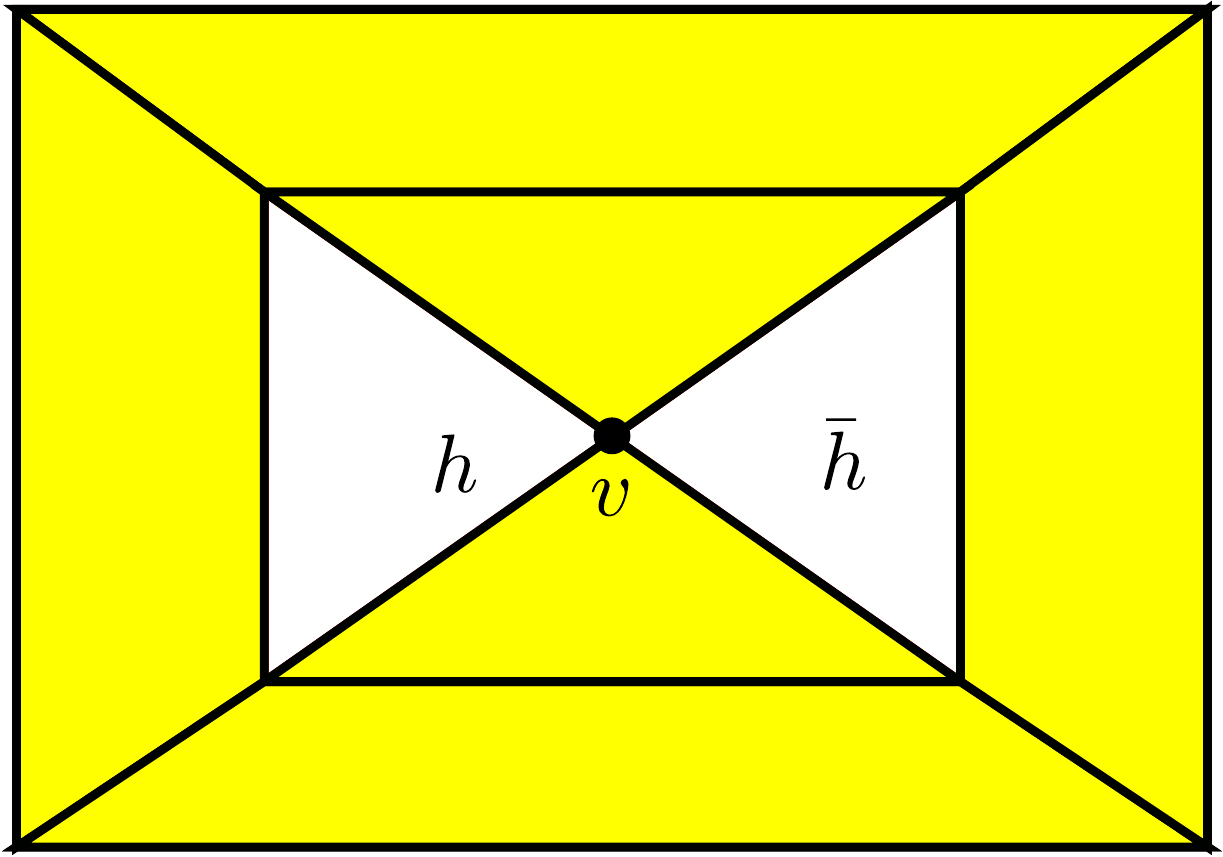}
      \quad
      \includegraphics[scale=0.3]{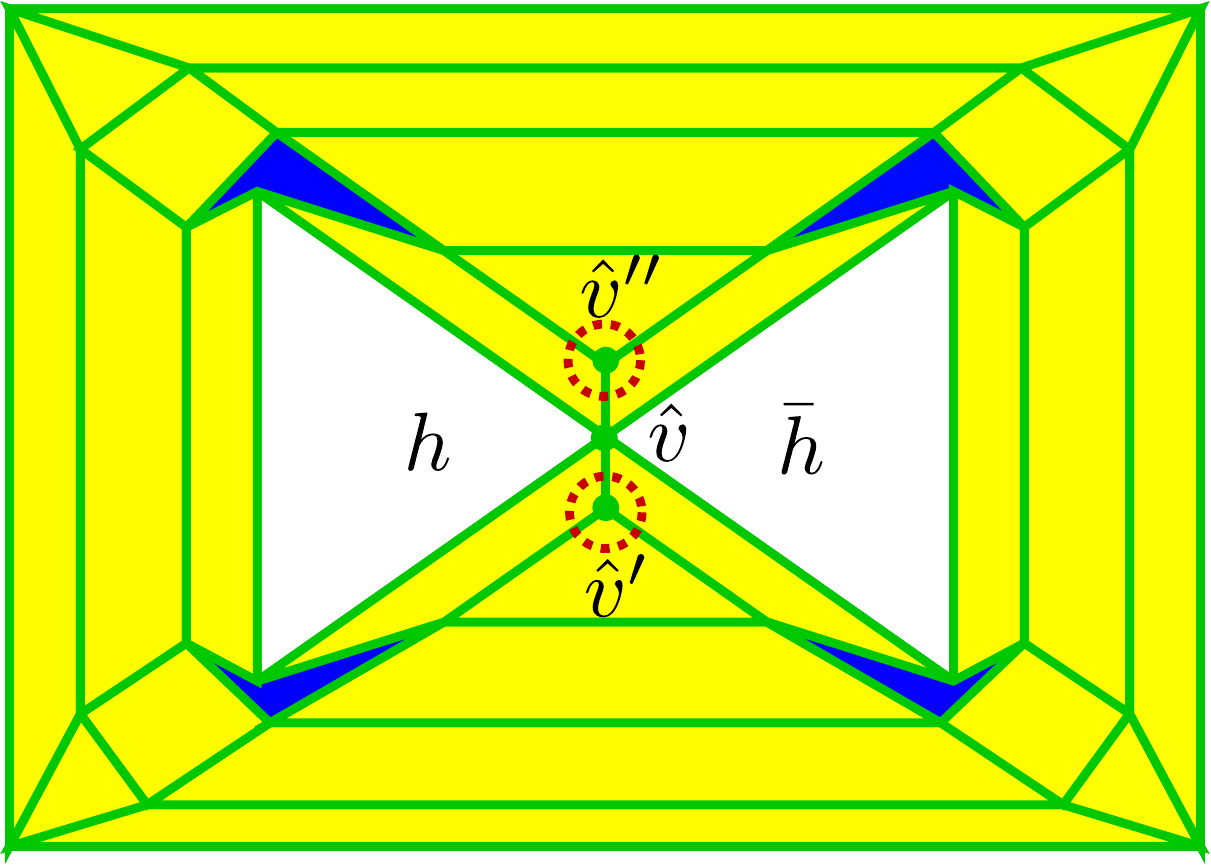}
      \caption{Two holes touching at a vertex (left), and the result of applying Euler transformation (right).
        Vertices with odd degree in the resulting complex are circled.
        The four cells shaded in blue are polygons of Class \ref{2dETcls3} generated by vertices in the input complex (see \cref{ssec:euler2d}).
        These cells could be nonconvex.
      }
      \label{fig:disjholes}
    \end{figure}
    
  }
\end{rem}

We pointed out in the Proof of \cref{lem:hGplanar} that the polygons of Class \ref{2dETcls2} in $\hK$ generated by edges are convex $4$-gons.
Each polygon $\hf \in \hK$  of Class \ref{2dETcls1} is geometrically similar to the polygon $f \in K$ generating it.
Hence if $f$ is convex, so is $\hf$.
But polygons of Class \ref{2dETcls3} generated by vertices are not guaranteed to be convex.
 In fact, when $v \in K$ is a boundary vertex where $K$ has a notch, or an ``incut corner'', $\hf_v \in \hK$ could be nonconvex---see \cref{fig:disjholes} for illustrations.

We finish this section with a result that guarantees $\hK$ remains connected.

\begin{prop}
  \label{prop:hKcnctd2d}
  Assuming the input complex $K$ in $d=2$ is connected, its Euler transformation $\hK$ is also connected.
\end{prop}
\begin{proof}
  We noted in the proof of \cref{lem:hGplanar} that the mitered offset polygons in $\hK$ are pairwise disjoint.
  But we show that when polygons $f,f' \in K$ are connected, so are the corresponding offset polygons $\hf, \hf' \in \hK$.
  By \cref{asmn:Kholesoutside} on the input complex, when polygons $f,f' \in K$ intersect, they do so either in an edge $e$ or in a vertex $v$.
  If $f \cap f' = e$, then by the definition of Euler transformation (\cref{ssec:euler2d}), the corresponding offset polygons $\hf,\hf' \in \hK$ are connected by the pair of new edges defining $\hf_e$, the $4$-gon of Class \ref{2dETcls2} generated by edge $e$.
  If $f \cap f' = v$ and $v$ is not an articulation vertex, then the corresponding offset polygons $\hf,\hf' \in \hK$ are similarly connected by the Class \ref{2dETcls3} polygon $\hf_v$ generated by $v$, with the corresponding copies $\hv, \hv'$ of $v$ in $\hf,\hf'$, respectively, being vertices of $\hf_v$.
  If $f \cap f' = v$ that is an articulation vertex, then $\hv=v$ is the identical copy of this vertex in $\hK$.
  There will be two Class \ref{2dETcls3} polygons $\hf_v, \hf_v'$ generated by $v$ in the two biconnected components joined at $v$, with $\hf_v \cap \hf_v' = \hv$.
  Further, $\hf_v$ is connected to $\hf$ and $\hf_v'$ to $\hf'$, ensuring that $\hf$ and $\hf'$ are connected.
  It follows that $\hK$ is connected, since we assume the input complex $K$ is connected.
\end{proof}

\subsection{Geometric Realization in $d=3$} \label{ssec:rlzn3d}

We restrict discussion of geometric realization in 3D to the case where $3$-cells in $\hK$ are homeomorphic to a $3$-ball and edges are straight lines, but $2$-cells can be non-planar.
For the sake of completeness, we relist the assumptions on the input complex $K$ in $d=3$ here.
\begin{asmn}
  \label{asmn:3dKCHCO}
  In dimension $d=3$, the input complex $K$, holes $\CH$, and the outside cell $\CO$ are assumed to satisfy the conditions specified in \cref{asmn:Kholesoutside}.
  In addition, we assume that the degree of each vertex $v$ in each $3$-cell $t \in K$ containing $v$ is $3$.
  In other words, each vertex $v \in t$ is connected to exactly three other vertices $v' \in t$.
\end{asmn}

It follows from \cref{asmn:3dKCHCO} that a vertex $v$ in a $3$-cell $t \in K$ is shared by exactly three polygons that are facets of $t$.
Tetrahedral,  cubical, and rectangular cuboid meshes are examples of polyhedral complexes satisfying the degree $3$ condition.

We first present the main result on the same even degree of vertices in the Euler transformation.

\begin{thm} \label{thm:deg6ET3d}
  Every vertex in $\hK$, the Euler transformation of the $3$-complex $K$, has degree $6$ in the $1$-skeleton of $\hK$.
\end{thm}
\begin{proof}
  Consider a vertex $\hv \in \hK$ that is part of the $3$-cell $\htt$, added as a Class \ref{3dETcls1} mitered offset copy of the $3$-cell $t \in K$ (see \cref{ssec:euler3d}).
  Let $v$ be the corresponding vertex in $t$.
  Since \cref{asmn:3dKCHCO} holds for $K$, vertex $v$ has degree $3$ in $t$ and is part of three polygons $\{f_i\}_{i=1}^3$ that are facets of $t$.
  As $\htt$ is a mitered offset of $t$, vertex $\hv$ has the identical degree of $3$ in $\htt$.
  Further, each $f_i$ generates a Class \ref{3dETcls2} $3$-cell $\htt_{f_i}$ in $\hK$ (see \cref{ssec:euler3d}) for $i=1,2,3$.
  Vertex $\hv$ is connected to one \emph{new} edge of the form  $\{\hv,\hv'_i\}$ in $\htt_{f_i}$ (disjoint from $\htt$) for $i=1,2,3$.
  These three edges bring the total degree of $\hv$ in the $1$-skeleton of $\hK$ to $6$.
\end{proof}
\begin{prop}
	\label{prop:transcomplex}
	Homology of $\hK$ do not change.
\end{prop}
Since $\hK$ and $K$ always have same underlying space.
\begin{rem}
  \label{rem:nonplnr3d}
  {\rm
  	While $3$-cells of each Class has a guaranteed geometric realization in $\R^3$ but some facets for $3$-cells belonging to Classes \ref{3dETcls3} and \ref{3dETcls4} can be non planar.
    %While $3$-cells of Classes \ref{3dETcls1} and \ref{3dETcls2} are guaranteed to have geometric realizations in $\R^3$, the same cannot be guaranteed for $3$-cells belonging to Classes \ref{3dETcls3} and \ref{3dETcls4}.
    In particular, if the cells in $K$ are assumed to be convex, then all Class \ref{3dETcls1} cells are also convex, since they are mitered offsets of the convex cells in $K$.
    In this case, the Class \ref{3dETcls2} cell $\htt_f$ generated by the convex polygon $f$ is also guaranteed to be convex, as each of the $p$ $4$-gon added is planar (see \cref{fig:3dETcls2}).
    
    The problem arises for the two $q$-gons added as part of each Class \ref{3dETcls3} cell---they may not be planar even if all edges $\he_j$ are identical in length.
    Instead of the depiction in \cref{fig:3dETcls3}, we might have the $q$-gon(s) as shown in \cref{fig:3dETcls3nplr}.
    We will find the best fit plane through all the vertices in question, and project the vertices up or down to the plane as needed by extending or shrinking edges as shown in the \cref{fig:3dETcls3nplr}
    %We can find a plane passing through subset of the points in  non-planar $q$-gon and project rest of the vertices of $q$-gon not in that plane by extending the edges as shown in the \cref{fig:3dETcls3nplr}. 
    %We might have to triangulate the $q$-gon for a geometric realization in this case.
    \begin{figure}[htp!]
      \centering
      \includegraphics[scale=0.2]{euler_transformation_3d_around_edge_input-crop.pdf}
      \hspace*{0.1in}
      \includegraphics[scale=0.2]{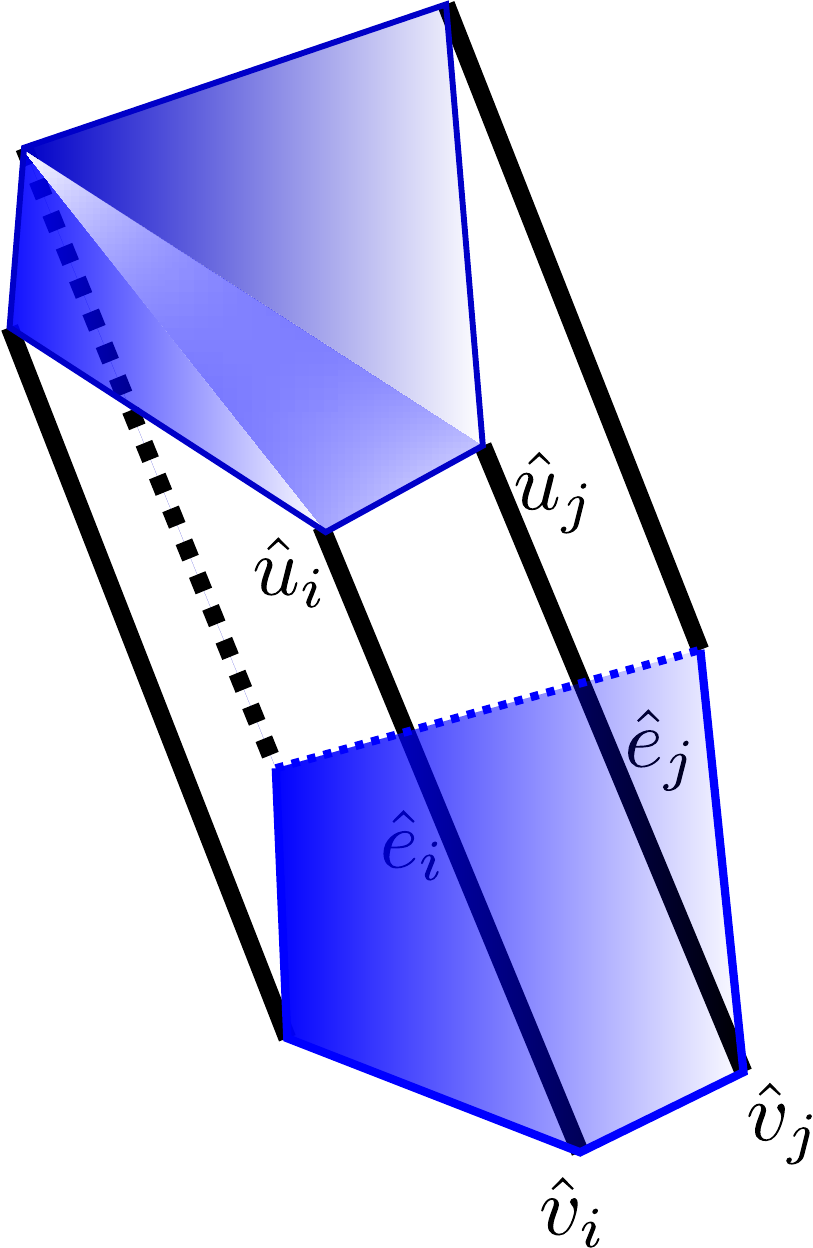}
      \hspace*{0.1in}
      \includegraphics[scale=0.2]{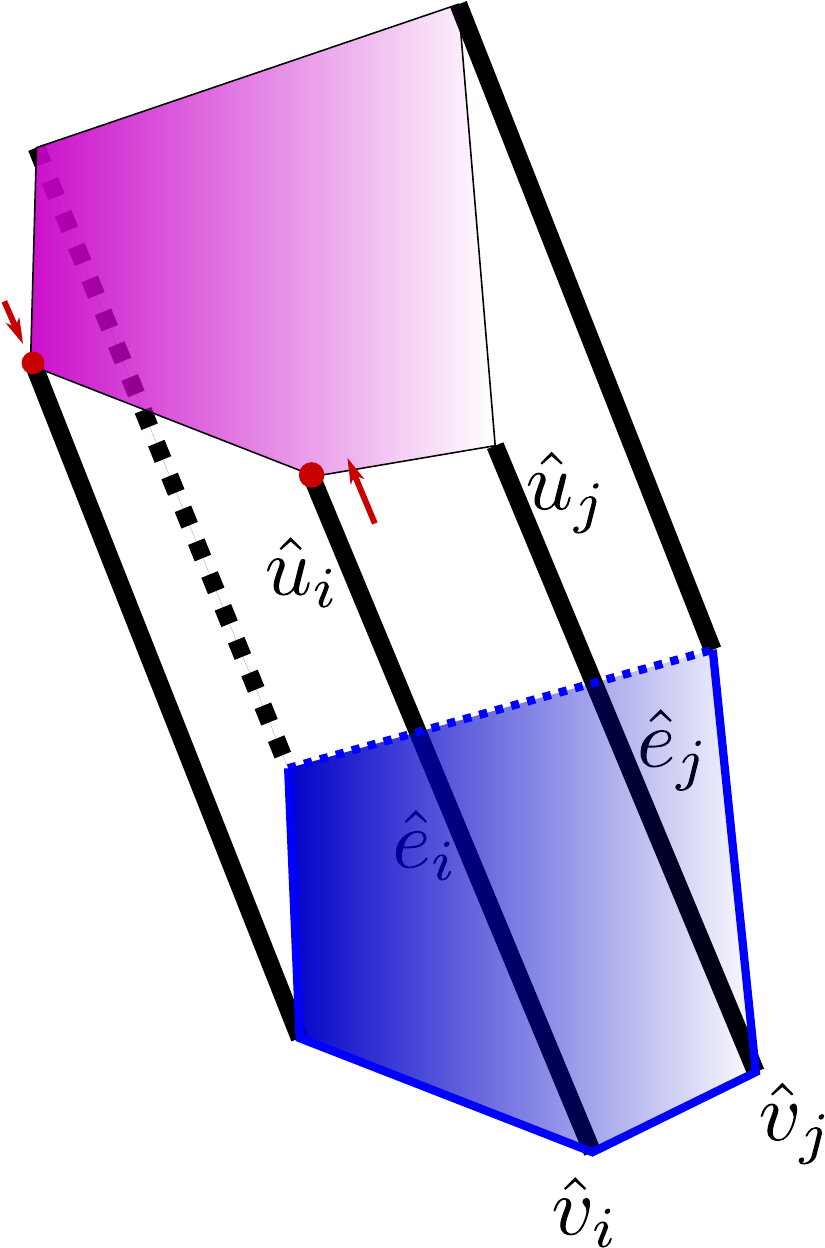}
      \caption{
        Input edge $e \in K$ shared by $5$ cells (left), just as shown shown in \cref{fig:3dETcls3}.
        But here, the pentagon on top in the Class \ref{3dETcls3} $3$-cell $\htt_e$ (middle) is not planar as originally set.
        This pentagon is replaced by the best fit plane (right, in pink) where one vertex is pushed up and another one down (red dots) while other vertices remain at their original positions. 
      }
      \label{fig:3dETcls3nplr}
    \end{figure}

    If the input complex is highly regular, e.g., a cubical complex with nearly identical cell sizes, it might be possible to choose the mitered offsets of each cell such that these polygons are indeed planar.
    But if such choices do not exist, we can project these points to a plane as discussed. 
    %But if such choices do not exist, we could consider triangulating each $q$-gon in such a way that the degree of each vertex $\hv_i$ stays even.
    %To this end, we could consider all candidate polygons that are facets of the Class \ref{3dETcls4} cell generated by vertex $v$, and add edges in their triangulations in such a fashion that they form cycles, thus ensuring we increase the degree of each vertex $\hv_i$ by even numbers.
    %But it is an open problem if we can always construct such triangulations.

    At the same time, for our motivating applications including infill lattice printing in additive manufacturing and coverage problems in robotics, we are concerned only with the $1$-skeleton of $\hK$.
    And the $1$-skeleton on $\hK$ is completely determined by the $3$-cells of Classes \ref{3dETcls1} and \ref{3dETcls2}.
    Hence we could just \emph{not} include $3$-cells of Classes \ref{3dETcls3} and \ref{3dETcls4} in $\hK$.
    Naturally, the underlying space $|\hK|$ will not be homeomorphic to $|K|$ because of the missing $3$-cells (we can show that $|K|/|\hK|$ is a single enclosed void).
    But the $1$-skeletons cover the same domain.
  }
\end{rem}

We now present bounds on the numbers of each class of cells in $\hK$ as multiples of corresponding numbers in $K$.
The counts would be lower if we do not include $3$-cells of Classes \ref{3dETcls3} and \ref{3dETcls4}.

\begin{lem}
  \label{lem:cntshThVhEhf3d}
  Let $V,E,F,T$ denote the sets of vertices, edges, polygons (faces), and $3$-cells in $K$, and let $\hV, \hE, \hF, \hT$ denote the corresponding sets in $\hK$.
  Let $\nu$ denote the maximum number of vertices in any $3$-cell in $K$, and let $\pi$ denote the maximum number of edges in any polygon (face) in $K$.
  The following bounds hold for the cardinalities of the sets of cells in $\hK$.
  \begin{align}
    |\hT| &   =  |T| + |F| + |E| + |V| \label{eq:3dhTsz} \\
    |\hV| & \leq \nu |T|  \label{eq:3dhVbd} \\
    |\hE| & \leq 3 \nu |T| \label{eq:3dhEbd} \\
    |\hF| & \leq (\pi+2)|F| + 2 |E| \label{eq:3dhFbd}
  \end{align}
\end{lem}
\begin{proof}
  \cref{eq:3dhTsz} follows from the definition of Euler transformation in $d=3$ in \cref{ssec:euler3d}, where we get a unique $3$-cell in $\hK$ corresponding to each $3$-cell (Class \ref{3dETcls1}), face or polygon (Class \ref{3dETcls2}), edge (Class \ref{3dETcls3}), and vertex (Class \ref{3dETcls4}) in $K$.
  The count will be equal to $|T| + |F|$ if we do not include $3$-cells of Classes \ref{3dETcls3} and \ref{3dETcls4}.

  Any two $3$-cells of Class \ref{3dETcls1} (mitered offsets) in $\hK$ are disjoint, and all vertices in $\hK$ belong to one of these $3$-cells.
  Hence the total number of vertices in $\hK$ is the sum of the number of vertices in each $\htt \in \hV$.
  The bound in \cref{eq:3dhVbd} results from the fact that each $3$-cell $t \in T$ generates a unique mitered offset $3$-cell $\htt \in \hT$, and the maximum number of vertices in any $t$ is $\nu$.

  \cref{thm:deg6ET3d} specifies that the degree of each vertex $\hv \in \hV$ is $6$.
  Hence we get that $2|\hE| = 6|\hV|$, the sum of all degrees in $\hK$.
  Combining this result with the bound in \cref{eq:3dhVbd} gives the bound in \cref{eq:3dhEbd} (assuming we add all Class \ref{3dETcls3} cells).

  Every facet or polygon $f \in F$ with $p$ edges generates a Class \ref{3dETcls2} $3$-cell with $p+2$ polygons as facets.
  Note that this set includes the two copies of $f$ that are facets of mitered offset $3$-cells in $K$ which share $f$.
  We get two \emph{new} polygons from the Class \ref{3dETcls3} $3$-cells generated by each edge $e \in E$.
  The bound in \cref{eq:3dhFbd} now follows since the number of edges of any $f \in F$ is $\pi$.
\end{proof}

We point out that the bounds in \cref{eq:3dhVbd,eq:3dhEbd,eq:3dhFbd} are tight when $K$ is a cell complex with same type of cells.
For instance, if $K$ is a tetrahedral complex, these bounds are tight with $\nu=4, \pi=3$.
If $K$ is a cubical complex, we get tight bounds with $\nu=8, \pi=4$.

\medskip
\noindent We now show that $\hK$ remains connected even if we do not add $3$-cells of Classes \ref{3dETcls3} and \ref{3dETcls4}.

\begin{prop}
  \label{prop:hKcnctd3d}
  Assuming the input complex $K$ is connected in $d=3$, its Euler transformation $\hK$ is also connected even without including Classes $3$ and $4$.
\end{prop}
\begin{proof}
  Observe that any pair of mitered offset $3$-cells $\hK$ (Class \ref{3dETcls1} cells specified in \cref{ssec:euler3d}) are disjoint by definition.
  But we show that if $K$ is connected, then so is $\hK$ even if we do not include $3$-cells of Classes \ref{3dETcls3} and \ref{3dETcls4}.
  By \cref{asmn:3dKCHCO}, when $3$-cells $t,t' \in K$ intersect in a polygon $t \cap t' = f$, then by the definition of Euler transformation (\cref{ssec:euler3d}), the corresponding offset polygons $\htt,\htt' \in \hK$ are connected by a new Class \ref{3dETcls2} $3$-cell $\htt_{f}$ generated by the polygon $f$.
  Also, by \cref{asmn:3dKCHCO}, for any two $3$-cells $t,t'$ in a biconnected component of $K$ with $t \cap t' = \emptyset$, there exists a finite sequence of $3$-cells $t_1=t, t_2,\dots,t_{r-1},t_r=t'$ such that $t_i \cap t_{i+1} = f_i$, a polygon, for $i=1,\dots,r-1$.
  Following the previous observation, each pair $\htt_i, \htt_{i+1}$ is connected in $\hK$ by a Class \ref{3dETcls2} $3$-cell $\htt_{f_i}$.
  Hence $\htt,\htt'$ are connected in the same biconnected component in $\hK$.
  Finally, articulation vertices are preserved by the Euler transformation, ensuring that biconnected components of $K$ remain connected in $\hK$.
  Hence it follows that $\hK$ is connected as well, since we assume the input complex $K$ is connected.
\end{proof}

We end this section with an additional nice property of the Euler transformed complex $\hK$: every edge is shared by exactly four polygons, i.e., faces.
As a result, the $2$-complex resulting from the intersection of $\hK$ with a plane in 3D that intersects only the edges but not any vertices in $\hK$ is guaranteed to be Euler.
This result could have implications in 3D printing where one may consider slicing the Euler transformed $3$-complex so as to print 2D layers.

\begin{lem}\label{lem:edgcnct4facs}
  Each edge in $\hK$ is connected to $4$ faces in $\hK$
\end{lem}
\begin{proof}
  We consider two cases for an edge $\he$ in $\hK$: $\he$ is part of some Class \ref{3dETcls1} $3$-cell $\htt$ and $\he$ is not part of any $3$-cell. %(see  \cref{ssec:euler3d} for details ).
  The two cases are illustrated on a cubical $3$-complex in \cref{fig:3dETedgefacerelation}.

  In the first case, $\he$ is the mitered offset of edge $e$ in $3$-cell $t \in K$.
  Let $e$ be shared by polygons $f',f''$ that are faces of $t$.
  Then $\he$ in $\htt$, the mitered offset of $t$, is shared by faces $\hf',\hf''$ of $\htt$, the mitered offset polygons of $f'$ and $f''$, respectively.
  Edge $\he$ is also shared by one polygon each from the two Class \ref{3dETcls2} $3$-cells $\htt_{f'}$ and $\htt_{f''}$, generated by $f'$ and $f''$, respectively.
  These two polygons are also faces of a Class \ref{3dETcls3} $3$-cell $\htt_{e}$ generated by $e$.
  Hence $\he$ is shared by exactly four polygons in $\hK$.

  %Each edge $\he$ is contained in $1$ class-$1$, $2$ class-$2$ and $1$ class-$3$ 3-cells. $\he$ is contained in $2$ faces $\hf, \hf'$ in $\htt$ (class - $1$). $\hf$ is contained in $1$ class-$2$ $3$-cell $t_f$ since every face in $\htt$ is also contained in some class-$2$ $3$-cell. Similarly $\hf'$ is contained in some class-$2$ $3$-cell $t_{f'}$. $t_f, t_{f'}$ contributes $2$ faces $\hf'', \hf'''$ that contains $\he$ since any $2$ same class $3$-cells are face disjoint. $\hf'', \hf'''$ are also contained in $1$ class-$3$ $3$-cell $t_e$ of $\hK$. Hence $\he$ is contained in $4$ faces.

  In the second case, let edge $\he = (\hv, \hv')$ be not contained in any Class \ref{3dETcls1} $3$-cell.
  By construction (see \cref{ssec:euler3d}), edge $\he$ is contained in one Class \ref{3dETcls2} and two Class \ref{3dETcls3} $3$-cells, and also one Class \ref{3dETcls4} $3$-cell separating these two Class \ref{3dETcls3} $3$-cells.
  Observe that no two $3$-cells of the same class in $\hK$ can share a polygon as a face.
  Let edges $e',e''$ meet at vertex $v$ in polygon $f$ in the input $3$-cell $t \in K$.
  Let the corresponding cells in the mitered offset $\htt \in \hK$ be $\he', \he''$ meeting at vertex $\hv$ in polygon $\hf$.
  Then edge $\he$ is contained in two $4$-gons that are faces of the Class \ref{3dETcls2} $3$-cell $\htt_{f}$.
  It is also contained in one polygon each from two Class \ref{3dETcls3} $3$-cells $\htt_{e'}$ and $\htt_{e''}$.
  Note that the last two polygons are also faces of the Class \ref{3dETcls4} $3$-cell $\htt_{v}$ generated by vertex $v$.
  Overall, $\he$ is shared by exactly four polygons in $\hK$.
  Note that We could, equivalently, describe the four polygons with respect to $\htt' \in \hK$, where the polygon $f$ is shared by $3$-cells $t, t' \in K$.
\end{proof}

\begin{figure}[htp!]
  \centering
  \includegraphics[scale=0.4]{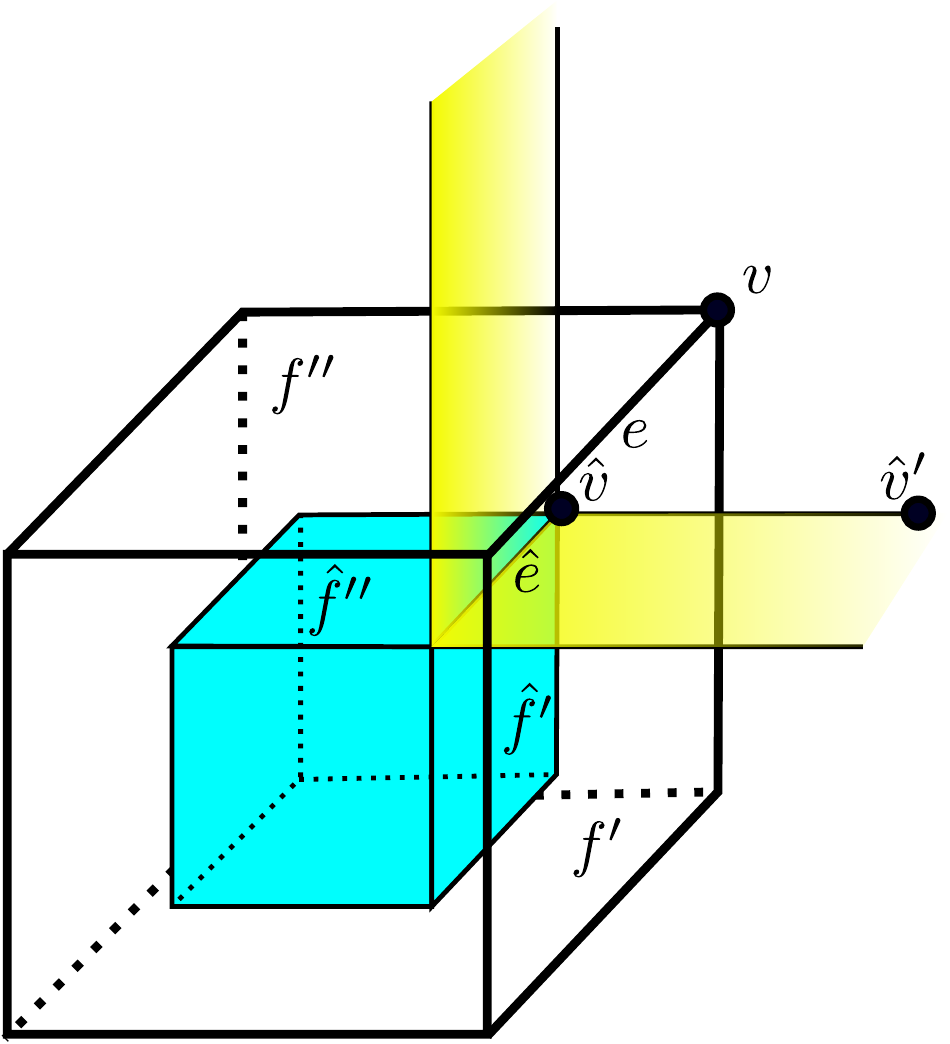}
  \quad
  \includegraphics[scale=0.4]{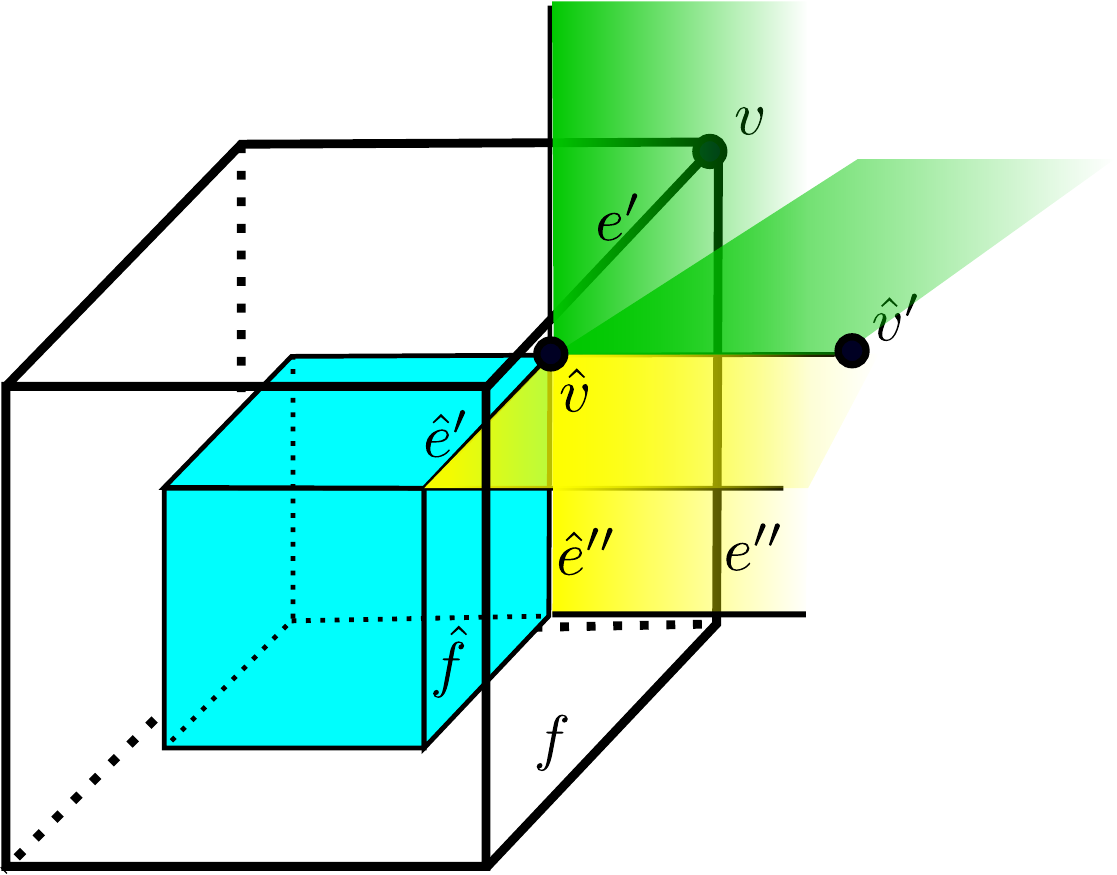}
  \caption{\label{fig:3dETedgefacerelation}
    Illustration of two cases for edge $\he$ in \cref{lem:edgcnct4facs} on a cubical complex.
    In the first case (left), the $3$-cell $t$ is the cube is shown in white, and its corresponding Class \ref{3dETcls1} $3$-cell $\htt \in \hK$ is shown in blue.
    In the right figure showing the second case, the two $4$-gons of the Class \ref{3dETcls2} $3$-cell $\htt_{f}$ are shown in yellow, and the polygons from two Class \ref{3dETcls3} $3$-cells are shown in green.\\
    }
\end{figure}

\begin{prop} \label{rem:2DsliceEuler}
  Let $\cP$ be a plane in $\R^3$ such that $\sigma \not\in \cP$ for every simplex $\sigma \in \hK$, the Euler transformed $3$-complex.
  Then the $1$-skeleton of the $2$-complex that is $\cP \cap \hK$ is Euler.
\end{prop}
\begin{proof}
  Since the plane $\cP$ does not contain any simplex of $\hK$, vertices in the resulting $2$-complex $\cP \cap \hK$ are created by $\cP$ intersecting edges in $\hK$.
  \cref{lem:edgcnct4facs} says each edge $e \in \hK$ is a face of exactly four polygons.
  Hence the vertex generated by $\cP \cap e$ will have degree $4$, with four edges generated by $\cP$ intersecting each of the four polygons that share $e$.
\end{proof}

Finally, we consider what happens if the restriction that each vertex has degree $3$ in each $3$-cell that contains the vertex in the input complex $K$ is not satisfied.

\begin{rem}\label{rem:3dETdeg4vtx}
  {\rm 
    As an example where the assumption of vertices with degree $3$ in each cell does not hold, consider a $3$-cell $t$ that is a pyramid with a trapezium at the base, as shown in  \cref{fig:mothde3ver}.
    Let vertex $v$ be the one on top with degree $4$ in the $1$-skeleton of $t$.
    Then the mitered offset $\ttt$ created after even an infinitesimal shrinkage will replace $v$ in $K$ with a face $\tf$ in the resulting complex $\tK$, instead of a corresponding single vertex $\hv$ \cite{AuWa2013,AuWa2016}.
    But since $v$ has degree $4$ in $t$, \cref{thm:deg6ET3d} is not valid, and we could get vertices with odd degrees.
    For instance, $\tv$ in \cref{fig:mothde3ver} has degree $5$.
    Nevertheless, every vertex in the $1$-skeleton of $\ttt$ will now have degree $3$.
    More generally, every vertex in the $1$-skeleton of $\tK$ obtained in this fashion has degree $3$, and it satisfies all requirements of \cref{asmn:Kholesoutside} of an input complex.
    If we apply Euler transformation \emph{again} on $\tK$, we are guaranteed to obtain a $3$-complex $\hK$ that satisfies all previous results presented in this Section.
    This situation is similar to the one in 2D described in \cref{rem:adjbdyedges}.
}	
\end{rem}
 
\begin{figure}[htp!]
  \centering
  \includegraphics[scale=0.45]{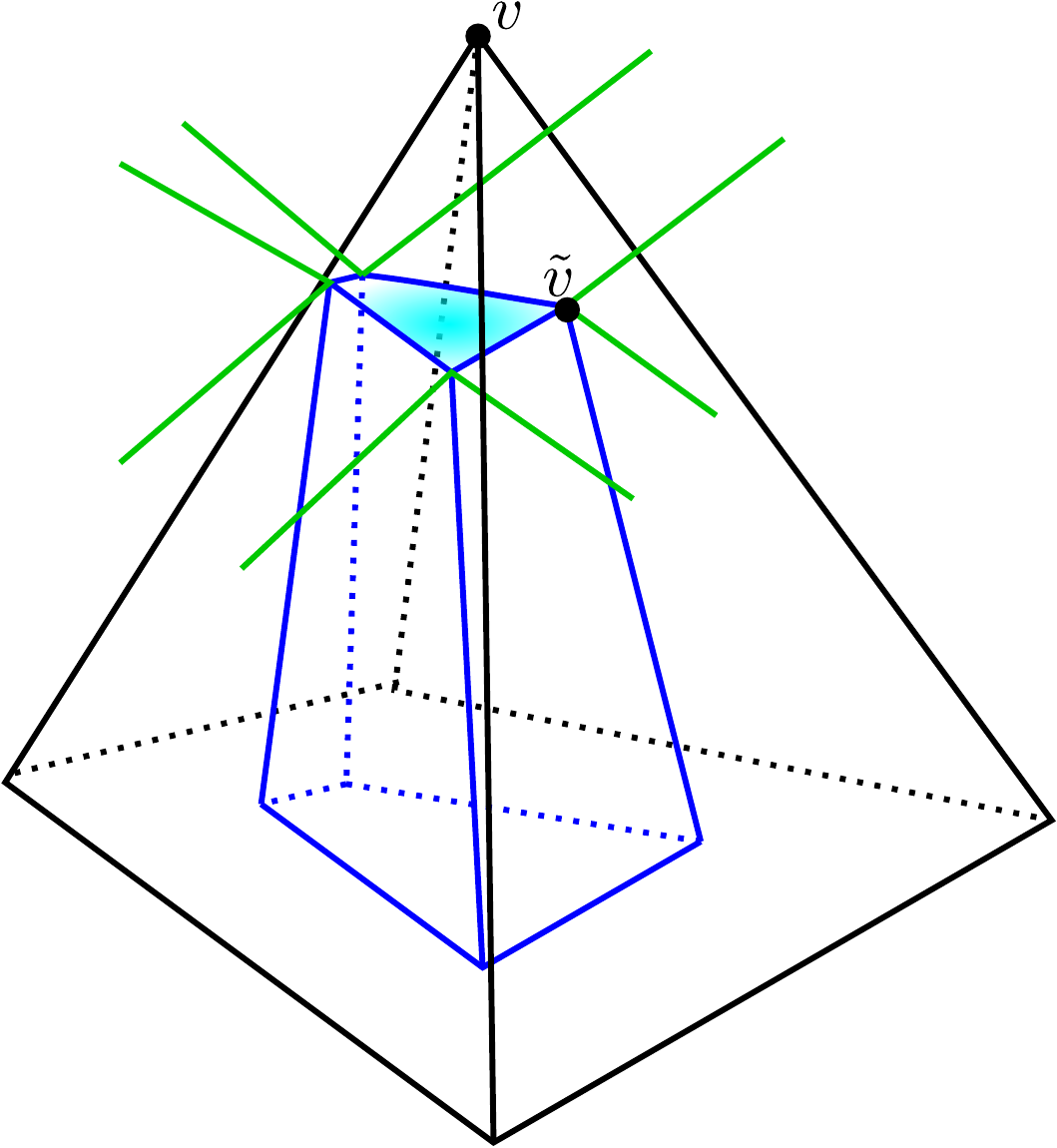}
  \caption{\label{fig:mothde3ver}
    A pyramid ($3$-cell) $t$ (black edges) in $K$, where vertex $v$ has degree $4$.
    The mitered offset of $t$ with infinite shrinkage $\ttt$ (blue edges) has vertices with degree $5$ in $\ttt$, e.g., $\tv$.
    But any vertex $\tv$ in $\ttt$ has degree $3$ in its $1$-skeleton, and hence satisfies the input assumption for applying Euler transformation.}
\end{figure}

\section{Measures of Geometric Quality} \label{sec:geomqual}

We inspect measures of geometric quality of the cells in the Euler transformation, and study how they compare with corresponding measures for the input complex.
We show that the user can choose a small set of parameters for each top-dimensional cell ($2$- or $3$-cells) in order to control these measures of quality.

Several measures of element quality are used in various domains ranging from numerical analysis, to finite element methods, to computer graphics \cite{ChDeSh2012,GiRaBa2012,LiEbZh2009,Si2010}.
We concentrate mostly on an \emph{aspect ratio} as defined below, but also present results on minimum edge lengths and angles in the 2D case.
\begin{defn}
  \label{def:asptrto}
  For a given $d$-cell $g$, let $\diam(g)$ denote its \emph{diameter}, i.e., the largest Euclidean distance between any two points in $g$.
  Also, let $\rho(g)$ denote its \emph{inradius}, i.e., the radius of the largest $d$-ball inscribed in $g$.
  The \emph{aspect ratio} of the cell is defined as
  \begin{equation}
    \label{eq:asptrto}
    \g(g) = \frac{\diam(g)}{\rho(g)} \,.
  \end{equation}
\end{defn}
%Note that the diameter is twice the circumradius, and hence $g$ is twice the ratio of circumradius over inradius, another commonly used aspect ratio.

\subsection{Geometric quality in $d=2$} \label{ssec:qual2d}

We consider the three classes of polygons generated in the Euler transformation $\hK$ corresponding to polygons, edges, and vertices in $K$ (\cref{ssec:euler2d}).
Since a Class \ref{2dETcls1} polygon $\hf$ is a mitered offset of the corresponding polygon $f \in K$, the user can choose parameters that control how $\hf$ is scaled with respect to $f$.
We consider two parameters ($\gl, \mu$) that control how edge lengths scale, and the offset parameter $b$, which is the perpendicular distance that each edge is pushed in when offsetting.
We define $\gl,\mu$ such that $\gl |e| \leq |\he| \leq \mu |e|$ holds for each edge $e \subset f \in K$.
We first choose $\gl$ and $\mu$ based on conditions given in \cref{eq:lammubtwn01,eq:edmamiratio}.
We denote by $|e_{\min}|, |e_{\max}|$ the minimum and maximum edge lengths in $f$.
These conditions ensure that a feasible value for $b$ can be chosen.
Based on $\gl$ and $\mu$ chosen, \cref{eq:edofbound} specifies a range of feasible values for $b$.
See \cref{fig:2dETcls1edgeconstr} for illustration of the angles used.
Choosing a $b$ in this range guarantees there are no topological or combinatorial changes in the polygon when offsetting.
We specify these ranges for \emph{each} polygon in terms of measures associated with $f$, i.e., input data.
The user could choose a uniform set of values over the entire complex, but the individual ranges afford more flexible choices.

These parameters are specific to $f$, but we do not use $\gl_f,\mu_f$ and so on in order to keep notation simpler.
We illustrate the constructions in \cref{fig:2dETcls1edgeconstr}.
Intuitively, the user may want to choose $\gl$ not too large and $\mu$ not too small in order to get good quality measures (bounded aspect ratio, minimum edge length, or maximum interior angle).
Let $b$ denote the offset distance for edge $e=\{u,v\}$, and let $r_u,r_v$ the corresponding distances for $u,v$.
By properties of mitered offset, choosing an edge length ratio (in $(\gl,\mu)$) and $b$ determines $r_u$ and $r_v$.
We denote by $p,q$ the projection lengths of $r_u, r_v$, respectively, on edge $e$.
Note that $p = b / \tan(\theta_u/2)$ and $q = b / \tan(\theta_v/2)$.

\begin{figure}[htp!]
  \centering
  \includegraphics[scale=0.3]{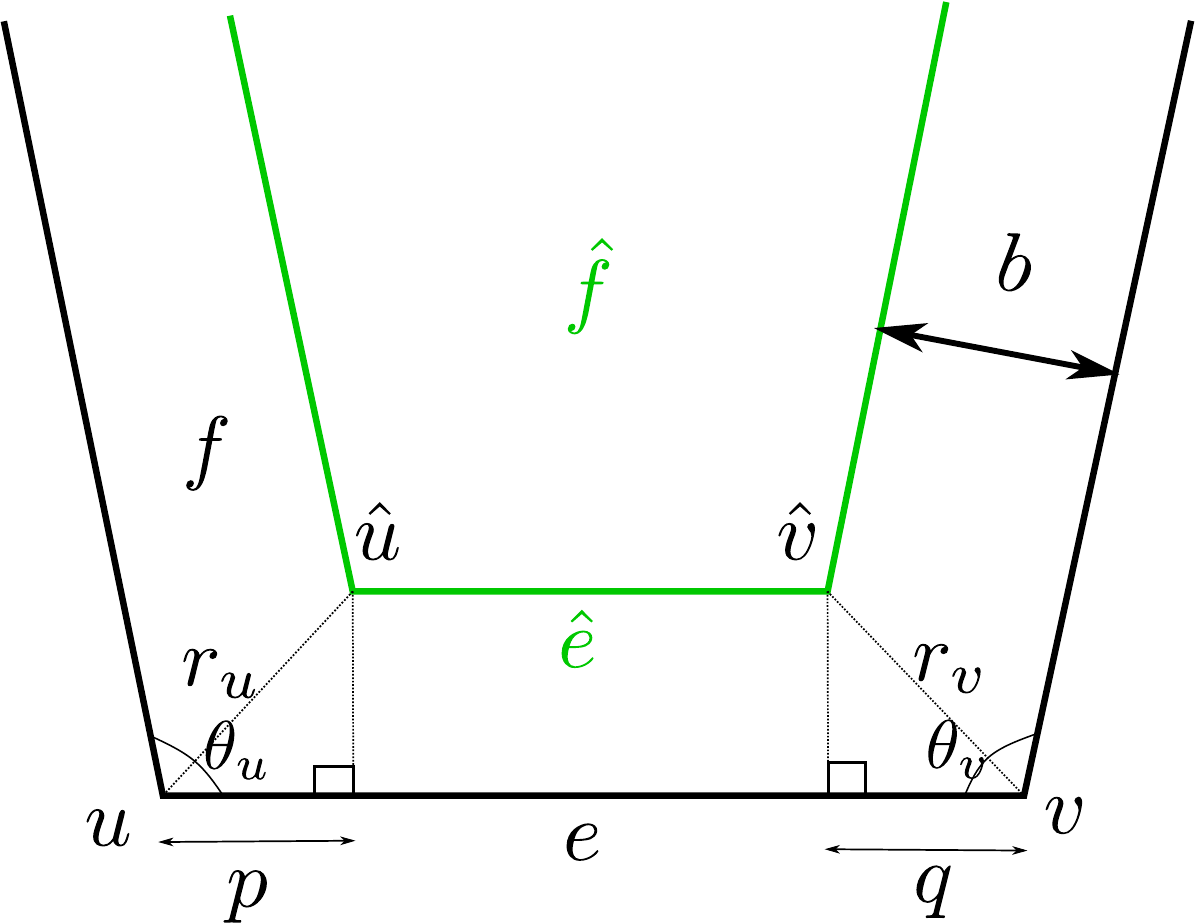}
  \caption{Offset polygon $\hf \in \hK$ (in green) and input polygon $f \in K$ (in black), with associated offset parameters.}
  \label{fig:2dETcls1edgeconstr}
\end{figure}

\begin{align}
  0 ~ < \gl & ~< \mu ~< 1 \, , \text{ and } \label{eq:lammubtwn01} \\
  \frac{|e_{\max}|}{|e_{\min}|} & ~ < ~ \frac{1- \lambda}{ 1- \mu} \, .  \label{eq:edmamiratio}  \\
  \hspace*{-0.2in} \max\{1, \tan(\theta_u/2),  \tan(\theta_v/2)\} \frac{(1 -\mu)|e|}{2} & ~< ~b~ < \min\{1, \tan(\theta_u/2),  \tan(\theta_v/2)\} \frac{(1 -\gl)|e|}{2}.  \label{eq:edofbound} 
\end{align}

Note that \cref{eq:lammubtwn01,eq:edmamiratio,eq:edofbound} guide the choices of $\gl, \mu$ and $b$ that the user can pick.
Once $\gl, \mu, b$ are chosen, we obtain certain bounds that are implied by these choices.
\cref{eq:edofbound} gives that $(1 - \mu) |e|\leq p + q \leq (1 -\gl) |e|$.
We have $\mu |e|\geq \he  \geq \gl |e|$, where $\he = |e| - (p + q)$.
Hence we get the tightest bounds on $b,p,q$ as
\begin{align}
  \frac{(1 -\mu)|e_{\max}|}{2} & < b, p, q < \frac{(1 -\lambda)|e_{\min}|}{2}.\label{eq:edofbound2}
\end{align}
\cref{eq:edofbound2} implies the following bounds on $r_u, r_v$:
\begin{align}
  \frac{(1 -\mu)|e_{\max}|}{\sqrt{2}} & < r_u, r_v < \frac{(1 -\lambda)|e_{\min}|}{\sqrt{2}}. \label{eq:vrofbound}
\end{align}  

We now specify bounds on quality measures for polygons in each of the three classes.

\subsubsection{Class \ref{2dETcls1} cells} \label{sssec:2dETcls1gq}
Let $f$ be a Class \ref{2dETcls1} polygon, and $b$ the edge offset distance for $\hf$.
Let $D, \hD$ be the diameters and $\rho, \hrho$ the inradii of $f \in K$ and $\hf \in \hK$, respectively.
See \cref{fig:2dETcls1gq} for details. 
\begin{figure}[htp!]
  \centering
  \includegraphics[scale=0.25]{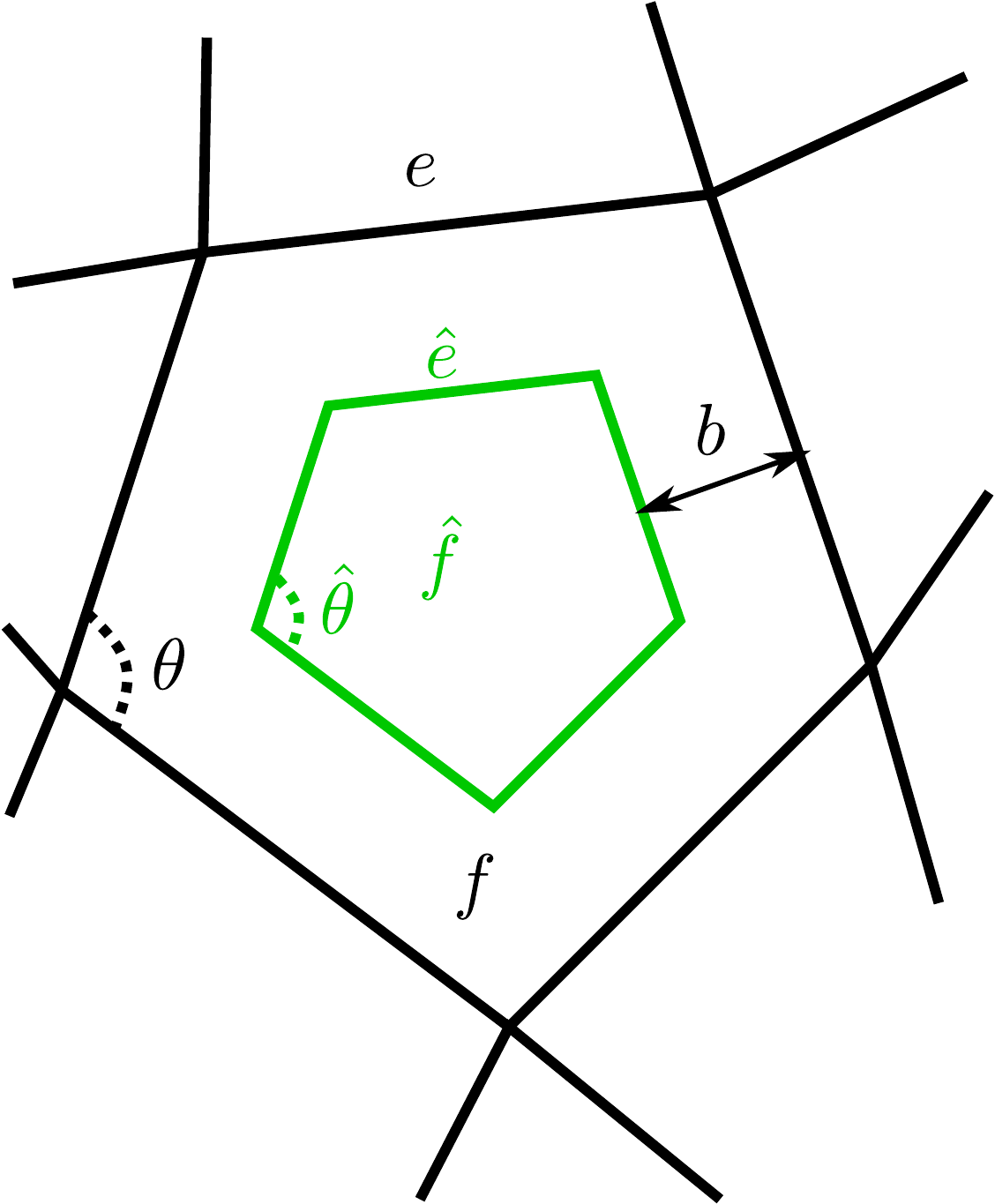}
  \caption{Quality measures for a Class \ref{2dETcls1} polygon $\hf$.}
  \label{fig:2dETcls1gq}
\end{figure}

\begin{description}
\item[Bounded Aspect Ratio]\label{bas:2dETcls1}
We know that $D \geq \hD$ and $\hrho \geq \rho - b$, and hence we get 
\begin{equation} \label{eq:2dETcls1ar}
    2 < \ghf \leq \frac{D}{\rho(1 - \frac{b}{\rho})} =\frac{\g(f)}{(1 - \frac{b}{\rho})}  \,.
\end{equation}
Since offset distance $b$ cannot be more than the inradius $\rho$, we have  $\frac{b}{\rho} < 1$.

\item[Minimum Edge Length]\label{mel:2dETcls1}
\begin{equation}
    \mbox{Since } |\he| \geq \gl |e| ~\forall e \in K, \mbox{ we get } |\he_{\min}| \geq \gl |e_{\min}| \,.
\end{equation}
\item[Maximum Interior Angle]\label{mia:2dETcls1}
  Since we are using mitered offsets, any interior angle  $\hat{\theta}$ in $\hf$ is same as the corresponding angle $\theta$ in $f$.
  Hence the maximum interior angle in $\hf$ is same as that in $f$.  
\end{description}

\subsubsection{Class \ref{2dETcls2} cells} \label{sssec:2dETcls2gq}
Let $r, r'$ be the offset distances for vertices $\hv, \hv'$ generated by $v \in e \in f, f' \in K$, and let $\he, \he'$ be the corresponding edges in $\hf, \hf'$, respectively.
Let $d$ be a diagonal in the Class \ref{2dETcls2} polygon $\hf_e$.
As in the previous Section, we let $D, \hD$ be the diameters and $\rho, \hrho$ the inradii of $f, \hf_e$, respectively.
We refer to \cref{fig:2dETcls2gq} for details.
\begin{figure}[htp!]
      \centering
      \includegraphics[scale=0.25]{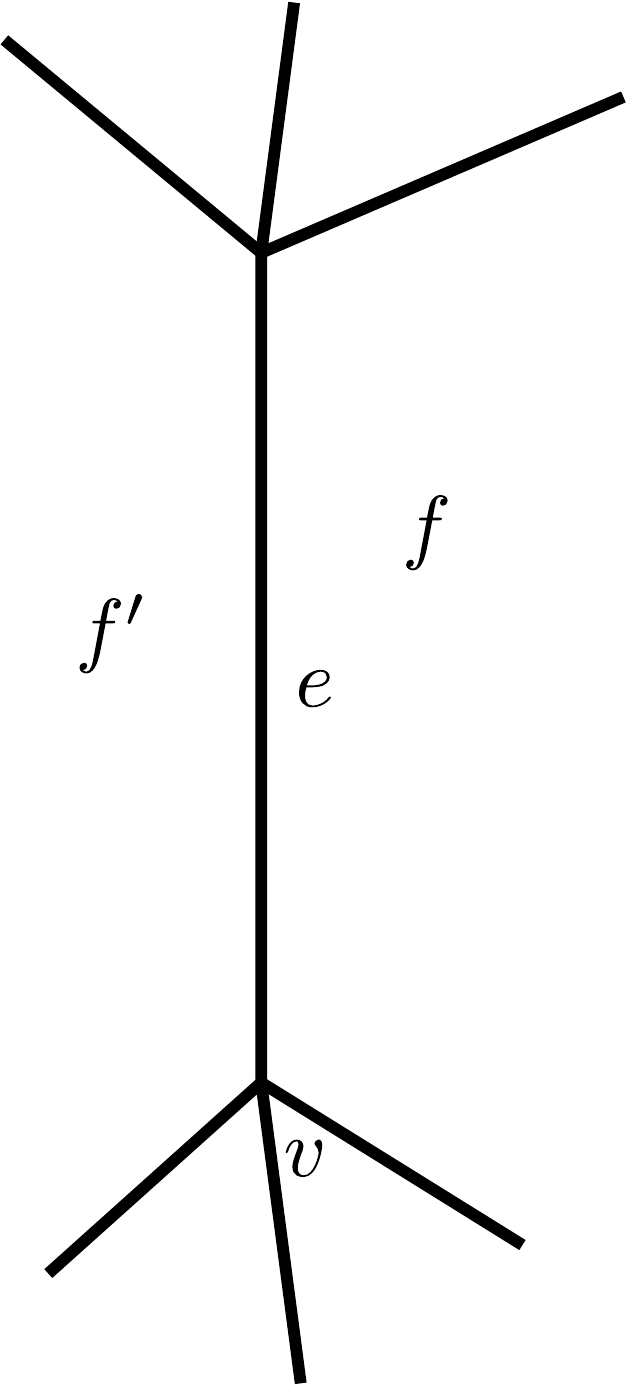}
      \quad\quad
      \includegraphics[scale=0.30]{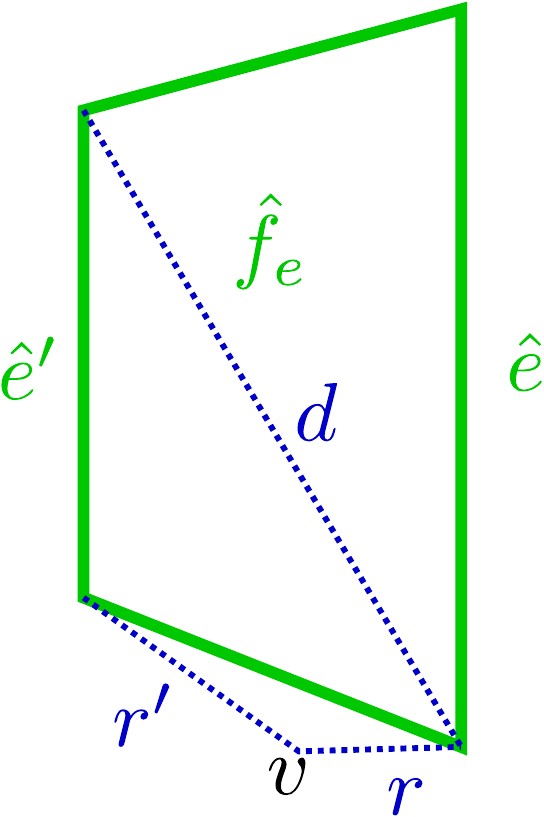}
      \quad\quad
      \includegraphics[scale=0.30]{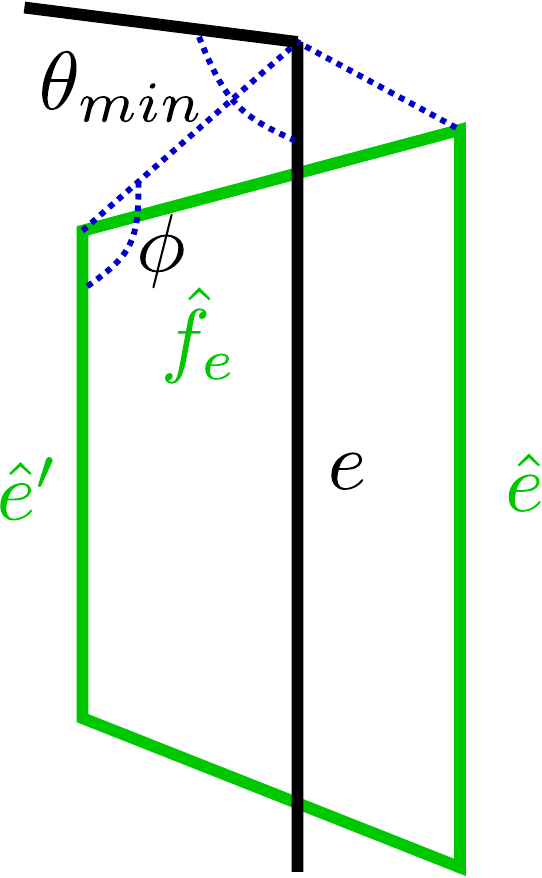}
      \caption{Quality aspects of a Class \ref{2dETcls2} cell $\hf_e$.
      }
      \label{fig:2dETcls2gq}
\end{figure}

\begin{description}
\item[Bounded Aspect Ratio]\label{bas:2dETcls2}
By triangle inequality, we get $|\he'| + r' + r \geq d$. 
Since $|e| \geq |\he'|$, we get
\[ |e| + r' + r \geq d \,. \]
Assume without loss of generality that $\lambda > \lambda'~,~ \mu > \mu'$  and  $|\he'| < |\he|$.
Then by \cref{eq:vrofbound}, we get
\begin{equation}
  ((1 - \lambda')\sqrt{2} + 1)|e| > d\,.   
\end{equation}
And since diameter of $\hf_e$ is not smaller than the maximum edge length and maximum diagonal length of a trapezium, we have $\hD < ((1 - \lambda')\sqrt{2} + 1)|e|$.
Using \cref{eq:edofbound}, and since $\mu > \mu' \text{ and } |\he'| < |\he| $ we get 
\begin{equation}\label{eq:bas:2dETcls2edofset}
  b, b' > 0.5(1 - \mu)|\he'| \,.   
\end{equation}
Then area of $\hf_e$ satisfies $\area(\hf_e) \geq (1 - \mu) |\he'|^2$.
As we know that $|\he'| \geq \lambda' |e| $, we get
\begin{equation}
    \area(\hf_e) \geq (1 - \mu) (\lambda')^2 |e|^2 \, .
\end{equation}
We know that the perimeter of $\hf_e$ satisfies $\per(\hf_e) \leq 4|e|$.
Since $\area(\hf_e) \leq \per(\hf_e) \hrho$ as $\hf_e$ is a trapezium \cite{ScAw2000}, we get
\begin{equation}
  \hrho > \frac{(1 - \mu ) (\lambda')^2 |e|}{4} \,.
\end{equation}
Hence
\begin{equation}
  2 < \g(\hf_e) < \frac{4 ((1 - \lambda')\sqrt{2} + 1)}{(1 - \mu) (\lambda')^2}\, .    
\end{equation}

\item[Minimum Edge Length]\label{mel:2dETcls2}
  Based on our assumption that $|\he'| < |\he|$, we have $\lambda' < \mu$.
  The other two non-parallel edges of $\hf_e$ have lengths at least $b + b'$.
  Hence the minimum edge length of $\hf_e$ is at least $\min \{|\he'|, b+b' \}$, which, by \cref{eq:bas:2dETcls2edofset}, is at least $\min\{ \lambda' |e|, (1 - \mu)\lambda' |e|\} = (1 - \mu)\lambda' |e|$.
  If $e$ is a boundary edge of $f$, the result remains same except $\lambda' = \lambda$ since $e$ is shared between $f$ and $\CH \cup \CO$. 

\item[Maximum Interior Angle]\label{mia:2dETcls2}
  Assume without loss of generality that the minimum interior angle of $f$ is smaller than that of $f'$.
  If $\theta_{\min}$ is the minimum interior angle of f, then the maximum interior angle $\hf_e$ is strictly less than $\phi =\pi - \theta_{\min}/2$ since sum of adjacent interior angles formed by the non-parallel lines with the parallel lines of a trapezium is $\pi$, and in the case of a mitered offset any line segment $u\hu$ shown in \cref{fig:2dETcls1edgeconstr} is an angle bisector for the interior angle at $u$ of the input cell.  
  This result is valid in both cases when $e$ is an interior or a boundary edge.
\end{description}

\subsubsection{Class \ref{2dETcls3} cells}  \label{sssec:2dETcls3gq}
Let $r, R'$ be the minimum and maximum offset distances for vertices $\hv, \hv'$, which correspond to vertex $v \in f, f'$, two specific polygons among all that share $v$.
See \cref{fig:2dETcls3gq} for details.
Let $|\tilde{e}_{\min}|, |\tilde{e}_{\max}|$ be the minimum and maximum edge lengths for edges in all $2$-cells sharing vertex $v$, and let the edges belong to $f', f$, respectively.
Let $\alpha, \beta $ are maximum and minimum angles formed by edges of $\hf_v$ on vertex $v$.
Also, let $\lambda, \mu$ and $\lambda', \mu'$ be the user-defined parameters for $f, f'$, respectively.

\begin{figure}[htp!]
      \centering
      \includegraphics[scale=0.25]{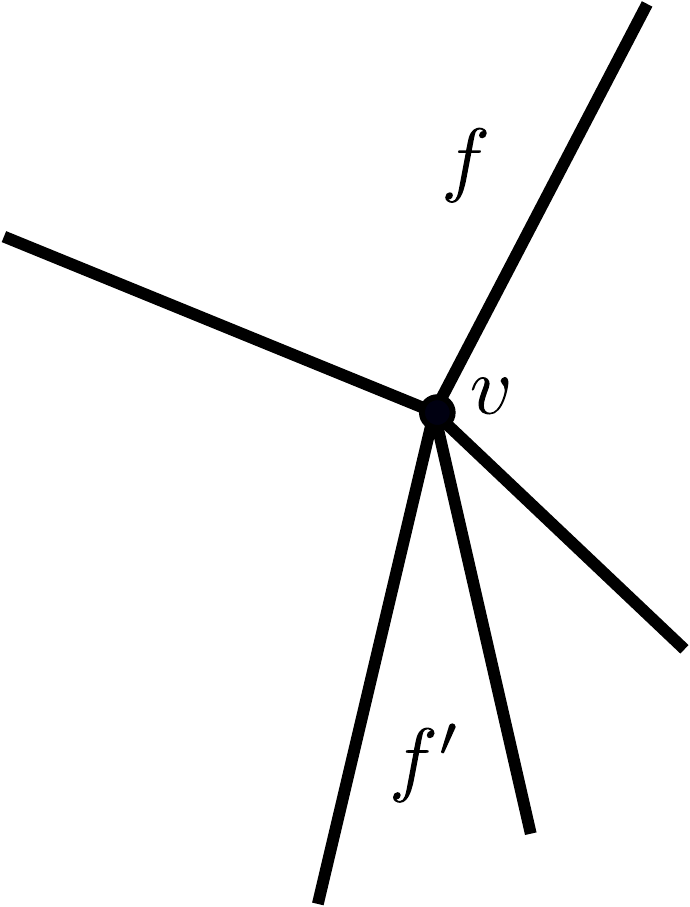}
      \quad
      \includegraphics[scale=0.25]{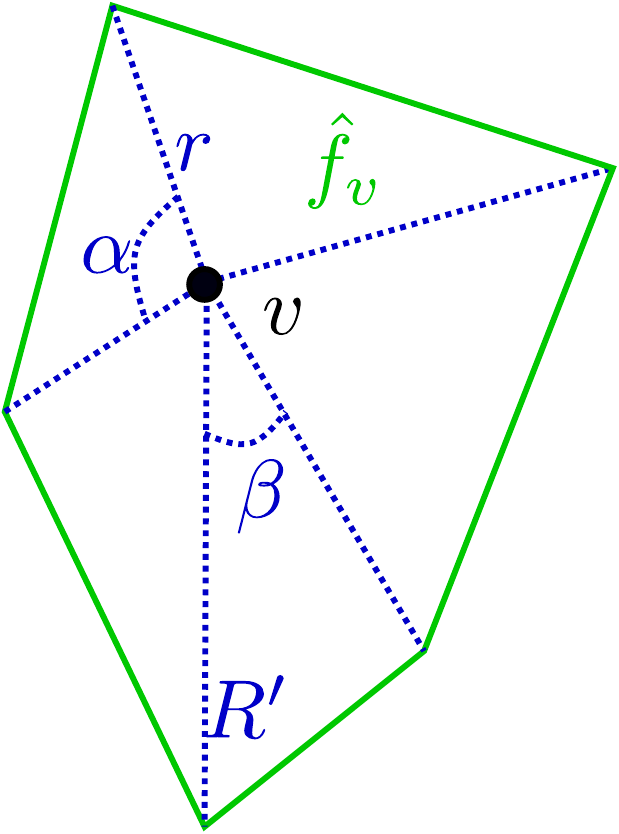}
      \quad
      \includegraphics[scale=0.25]{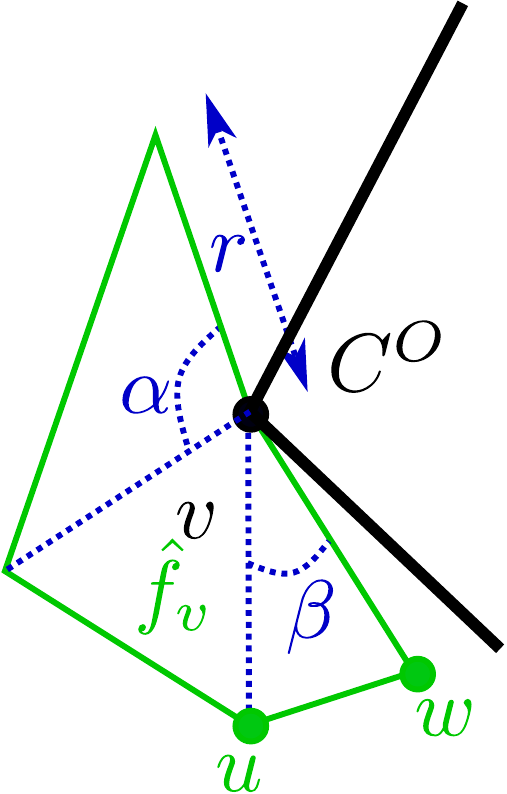}
      \quad
      \includegraphics[scale=0.25]{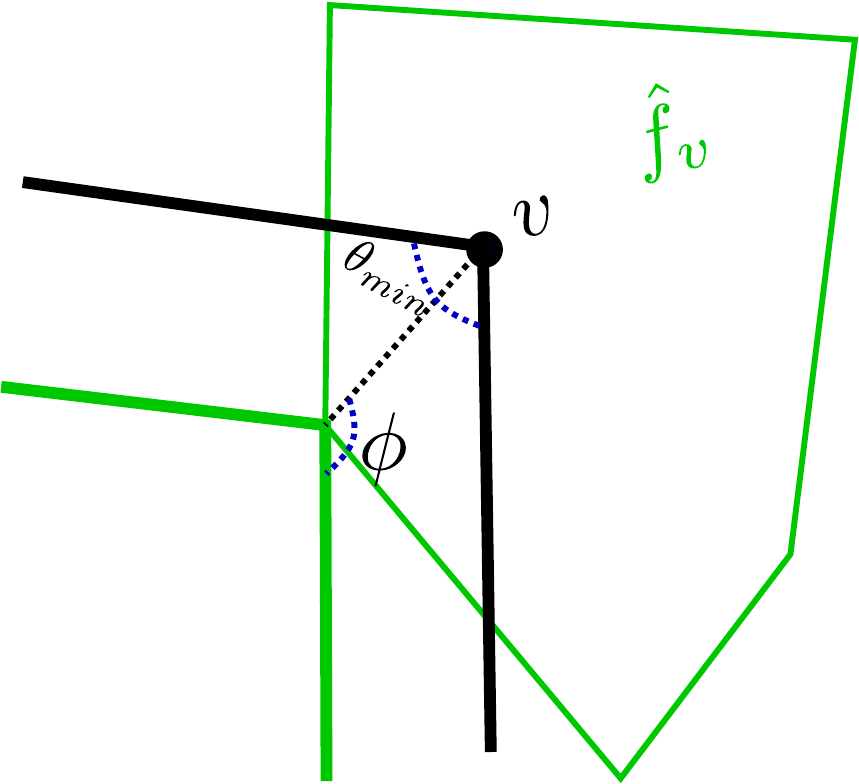}
      \caption{Quality measures for a Class \ref{2dETcls3} cell $\hf_v$.
      }
      \label{fig:2dETcls3gq}
\end{figure}

\begin{description}
\item[Bounded Aspect Ratio]\label{bas:2dETcls3}
  We first consider the case where $v$ is an interior vertex, i.e., it is not on the boundary of $K$.
  Since $\hrho \geq r \cos{\frac{\alpha}{2}}$ and $\hD \leq 2 R'$,  using \cref{eq:vrofbound} gives
  \begin{equation}
    2 < \g(\hf_v) \leq \frac{2R'}{r\cos{\frac{\alpha}{2}}} < C\bigg(\frac{1 - \lambda'}{1 - \mu}\bigg) 
  \end{equation}
  where $C = 2 |\tilde{e}_{\min}| /(|\tilde{e}_{\max}| \cos{\frac{\alpha}{2}})$.
  
  When $v$ is on the boundary, we consider applying the same offset $r$ to every $2$-cell sharing vertex $v$.
  Under this setting, $\Delta uvw$ is an isosceles triangle (third image in \cref{fig:2dETcls3gq}), and hence its inradius is $r \sin(\frac{\beta}{2}) \sqrt{(1- \sin(\frac{\beta}{2}))/(1 + \sin(\frac{\beta}{2}))}$.
  \begin{equation}
    \mbox{Denoting } L = \sin(\frac{\beta}{2}) \sqrt{(1- \sin(\frac{\beta}{2}))/(1 + \sin(\frac{\beta}{2}))},
    \mbox{ we get }~~~~~ \g(\hf_v)  < \bar{C}\bigg(\frac{1 - \lambda'}{1 - \mu}\bigg)~
  \end{equation}
  where, $ \bar{C} = 2|\tilde{e}_{\min}| / (|\tilde{e}_{\max}| L)\,$.

\item[Minimum Edge Length]\label{mel:2dETcls3}
If $v$ is not on the boundary, we get
\begin{equation}
    |\he| \geq 2 r \sin{(\frac{\beta}{2})} > \sqrt{2}(1 - \mu) |\tilde{e}_{\max}| \sin{(\frac{\beta}{2}}) \,.
\end{equation}
If $v$ is on the boundary, we get
\begin{equation}
    |\he| \geq \min\bigg\{ r, 2r\sin{\frac{\beta}{2}}\bigg\} \,.
\end{equation}
Then we get by \cref{eq:vrofbound} that
\begin{equation}
    |\he| > \sqrt{2}(1 - \mu)|\tilde{e}_{\max}| \min\bigg\{ \frac{1}{2}, \sin{\frac{\beta}{2}}\bigg\}\,.    
\end{equation}

\item[Maximum Interior Angle]\label{mia:2dETcls3}
  As shown in the Figure \ref{fig:2dETcls3gq}, $\phi < \pi - (\theta_{\min}/2)$, where $\theta_{\min}$ is minimum interior angle of the $2$-cells sharing vertex $v$.
  Hence maximum interior angle of $\hf_v$ is strictly less than $2\phi < 2\pi - \theta_{\min}$.  
\end{description}

\medskip
The results in \cref{sssec:2dETcls1gq,sssec:2dETcls2gq,sssec:2dETcls3gq} show that the user could choose parameters $\gl_f,\mu_f,b_f$ for each cell $f \in K$ so that all measures of geometric quality of cells in $\hK$ are within desired bounds.
In order to make $\hK$ as uniform as possible, the user could choose a single, or a few, set(s) of values for these parameters that are applied for all cells.
For instance, one could choose a $\gl^* \geq \gl_f$ and $\mu^* \leq \mu_f$ with $\gl^* < \mu^*$ for all $f \in K$.
But such choices may not exist for all input complexes $K$.

\subsubsection{Euclidean Bound on Length after Euler Transformation }\label{ssec:euclideanbound}
Based on geometric quality parameter discussed in section \ref{ssec:qual2d}. Let $\lambda*, \mu*$ are the parameters for input complex $K$ such that $\lambda^*|e|\leq |\hat{e}|\leq \mu^*|e|~~~\forall e \in E$ .  $b, p, q < \frac{|e|(1 - \lambda^*)}{2} \implies r, s < \frac{|e|(1 - \lambda^*)}{2}$ as shown in Figure \ref{fig:eucledian_bound-crop}.
$\hat{e} = \sqrt{s^2 + b^2} < (1 - \lambda^*)|e| \frac{\sqrt{5}}{2}$. For each edge $e$ we have unique $\hat{e}, \hat{e}'$ edges in $\hat{K}$, hence total length of these types of edges is $< \sum_{e \in E} (1 - \lambda^*)|e| \sqrt{5}$. Sum of all the edges of class $1$ in complex $\hat{K}$ is $\leq 2 \sum_{e \in E} \mu^* |e|$. Hence total euclidean length($\hat{L}$) of edges in $\hat{K}$ is 
$2L < \hat{L} < \sqrt{5}\sum_{e \in E}|e| + (2\mu^* - \sqrt{5}\lambda^*)\sum_{e \in E}|e| = (\sqrt{5} + (2\mu^* - \sqrt{5}\lambda^*))L$, where $L = \sum_{e \in E}|e|$    
$\hat{L}$ is at least going to be double as compared to euclidean length of edges of original complex $K$, but it also depends upon value of $\mu^*, \lambda^*$. We can control density with parameters $\mu^*, \lambda^*$ to some extent.
\begin{figure}[htp!] 
	\centering
	\includegraphics[width=70mm,height=50mm]{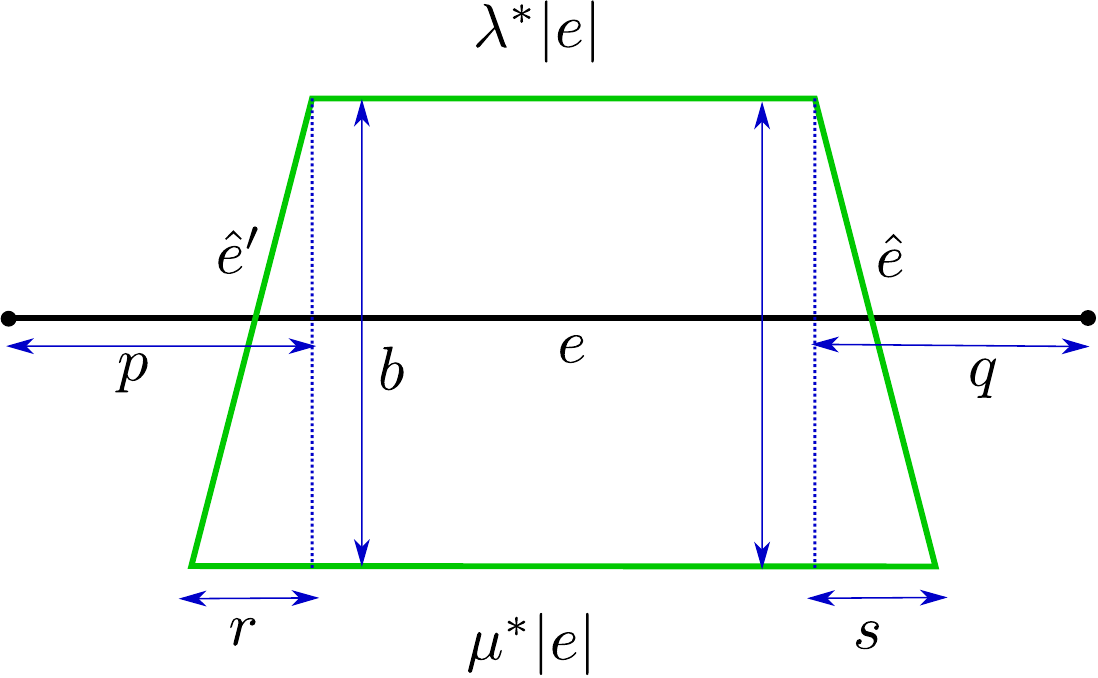}
	\caption{$4$-gons (green) is a class-$2$ $2$-cell in $\hat{K}$ created around the edge $e$ of input complex $K$.}
	\label{fig:eucledian_bound-crop}
\end{figure}

\subsection{Geometric quality in $d=3$} \label{ssec:qual3d}

We specify how the aspect ratios of $3$-cells of Classes \ref{3dETcls1} and \ref{3dETcls2} compare with those of cells in $K$.
Deriving similar bounds for quality measures of Class \ref{3dETcls3} and \ref{3dETcls4} cells appears more challenging, and is left as part of future work.
In $d=3$, we let the user specify the mitered offset distance $b$ by which each facet (polygon) of a $3$-cell $t \in K$ is moved in to create the corresponding offset cell $\htt \in \hK$.

\subsubsection{Class \ref{3dETcls1} cell}
Let $D, \hD$ denote the diameters, and $\rho, \hrho$ the inradii of cells $t, \htt$, respectively.
Since the Class \ref{3dETcls1} $3$-cell $\htt$ is created as a mitered offset of $3$-cell $t \in K$, we have $~\hD \leq D~$ and $~\hrho \geq \rho - b$.
Hence we get that
\[
\g(\htt) = \frac{\hD}{\hrho} < \frac{D}{\rho - b} = \frac{\g(t)}{\left( 1 - \frac{b}{\rho}\right)}, ~\mbox{ with } \frac{b}{\rho} < 1 \, .
\]
Note that this is the same bound satisfied by the aspect ratio of a Class \ref{2dETcls1} polygon specified in \cref{eq:2dETcls1ar}.

\subsubsection{Class \ref{3dETcls2} cell}
In general, the Class \ref{3dETcls2} $\htt_f$ generated by facet $f \in t \in K$ could have the shape of an ``oblique'' truncated cone with $\hf,\hf'$ as the bases (unlike the ``orthogonal'' truncated cone that \cref{fig:3dETcls2} might indicate).
let $\bar{D}=D+D'$ be the diameter of a ball that contains both $3$-cells $t,t'$ that share $f$ as a facet.
Let $\hat{r}, \hat{r}', r$ be the inradii  of facets $\hf, \hf', f$, respectively, and let $b, b'$ be the offset distance for $f, f'$.
Assume without loss of generality that $\hat{r} < \hat{r}'$.
We can fit an oblique cylinder of radius $\hat{r}$ and height $b + b'$, as shown in \cref{fig:3dETcls2gq}, inside $\htt_f$.
Then we can fit a ball of radius $R = \min\{ \frac{b+b'}{2}, \hat{r}\} \cos{\theta}$, where $\theta$ is the skew angle of the oblique cylinder.
Note that $\theta = 0$ if $\htt_f$ is an ``orthogonal'' truncated cone, in which case the cylinder will be normal as well.
We get
\[\g(\htt_f) = \frac{\hD}{\hrho} \leq \frac{\bar{D}}{R} \,.\]
\begin{figure}[htp!]
      \centering
      \includegraphics[scale=0.35]{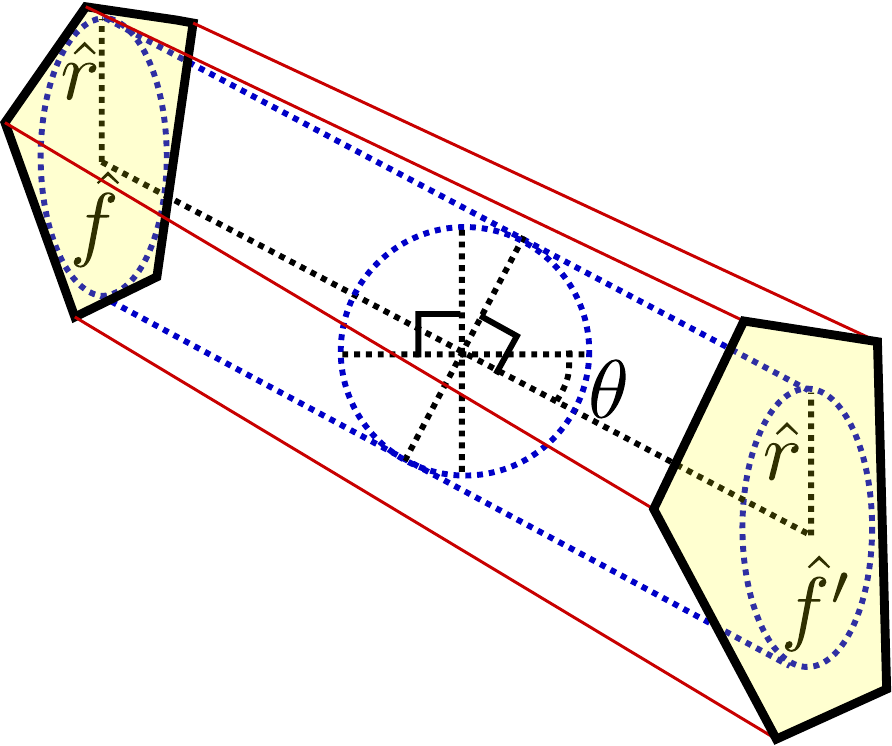}
      \caption{The oblique cylinder we could fit within a Class \ref{3dETcls2} cell $\htt_f$, and the parameters that determine its aspect ratio.
      }
      \label{fig:3dETcls2gq}
\end{figure}

\section{Application: Additive Manufacturing} \label{sec:3dprtg}

We have been using the Euler transformation for tool-path planning in large scale additive manufacturing (3D printing) with our collaborators in Oak Ridge National Laboratory (ORNL).
The system ORNL is developing is termed Big Area Additive Manufacturing (BAAM), and its goal is to be able to \emph{efficiently} print 3D objects that are much larger than the ones that typical 3D-printers in the market today are able to handle (think a few to several feet in each dimension).
Being able to print all edges in a contiguous manner is critical for efficiency.

We present brief details of a proof-of-concept print of a pyramid, which was represented as a collection of planar layers stacked from the ground up (see the left image in \cref{fig:3dprtg}).
The dimensions of the pyramid were $609.6\,$mm $\times$ $609.6\,$mm $\times$ $609.6\,$mm, and each layer had a height of $4.26\,$mm, resulting in a total of $143$ layers.
We started with a standard triangulation $K_0$ of the bottom layer, and obtained the Euler transformation $\hK_0$.
Given the geometry of this object, we took the intersection of $\hK_0$ with the domain of each subsequent layer $i>0$, trimming any vertices and edges of $\hK_0$ that were outside the layer's domain to create $\hK_i$.
This process created some odd-degree nodes at the boundary in $\hK_i$.
But we can show that the number of odd degree nodes in $\hK_i$ is even, and hence we added a set of edges along the perimeter of the layer connecting pairs of odd-degree nodes.
With this modification, $\hK_i$ is guaranteed to be Euler again.

For printing each layer, we used a blossom algorithm to identify an Eulerian tour.
Starting from any vertex, we choose edges at each step that minimize sharp turns in a greedy manner to identify the tool path.

The method outlined above naturally handles voids in the print domain.
We are developing various classes of efficient algorithms for additive manufacturing in this fashion that are capable of incorporating several physical and material constraints.
One such framework has been presented in a separate manuscript \cite{GuKrDr2020}.

\begin{figure}[htp!]
  \centering
  \includegraphics[width=0.45\textwidth]{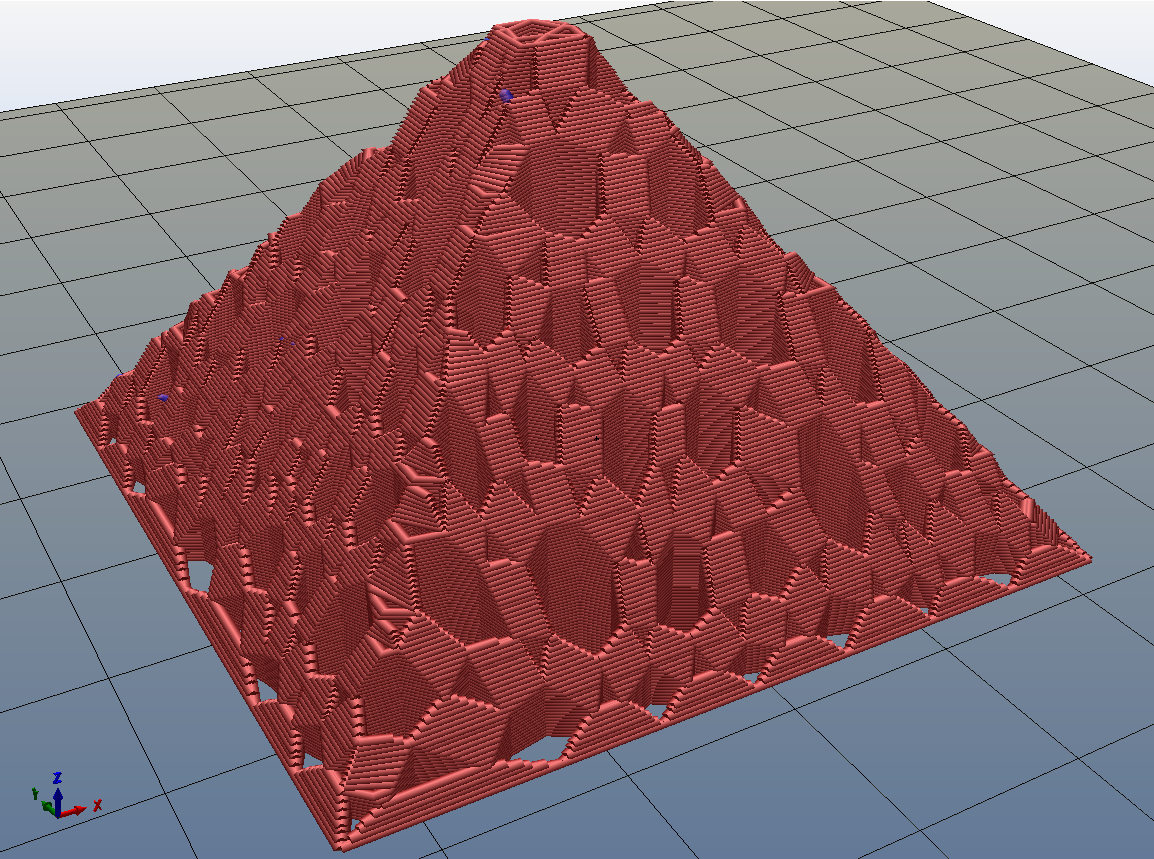}
  \quad
  \includegraphics[width=0.47\textwidth]{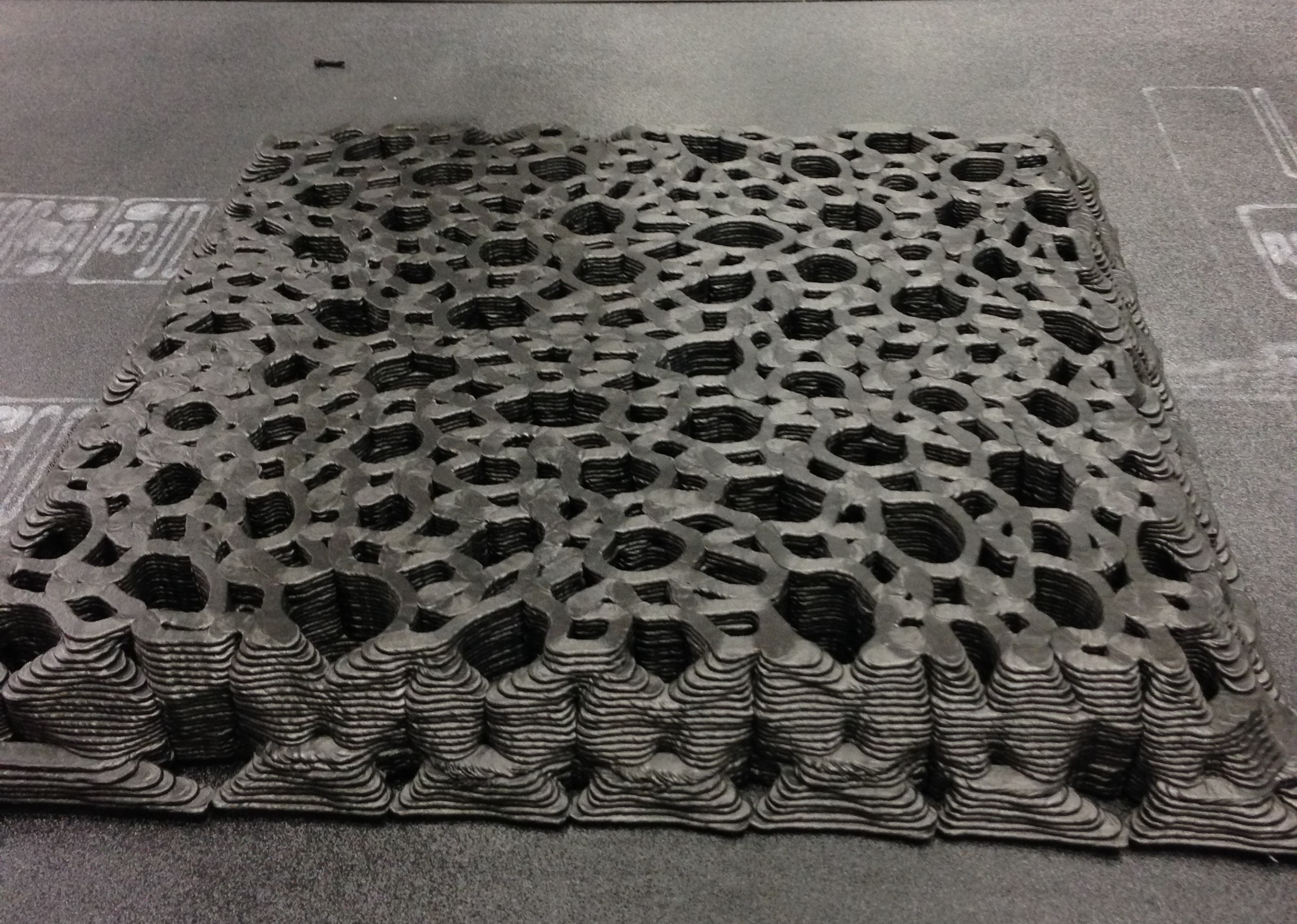}
  \caption{Visualization of the design of a pyramid-shaped object (left), and a view of the partially 3D-printed object (right). }
  \label{fig:3dprtg}
\end{figure}

\section{Discussion} \label{sec:disc}

The bottleneck for computational complexity of the Euler transformation is determined by the mitered offsets it creates for each cell in $K$.
The number of cells in $\hK$ are clearly linear in the number of cells in $K$ (\cref{lem:cntshVhEhF2d,lem:cntshThVhEhf3d}).
For $d=2,3$, if $K$ in $\R^d$ has $m$ $d$-cells, each of which has at most $p$ facets, the time complexity of Euler transformation is $O(m p^d)$ \cite{AuWa2013,AuWa2016}.

Not all cells in the Euler transformation $\hK$ are guaranteed to be convex, even when all cells in $K$ are (see \cref{fig:disjholes}).
We presented various results on the measures of geometric quality of each class of cell in $\hK$ (in \cref{sec:geomqual}).
At the same time, choosing an \emph{optimal} set of parameters ($\gl, \mu, b$) that maximizes the overall quality of \emph{all} classes of cells in $\hK$ simultaneously could be modeled as an optimization problem.
%We could triangulate the non-convex cells so that all cells in $\hK$ are convex.
%But could we do so while maintaining even degrees for all vertices?
%A related problem is that of finding a triangulation (rather than a cell complex) of a given domain that minimizes the number of odd-degree vertices.

We pointed out in \cref{rem:nonplnr3d} that Class \ref{3dETcls3} or \ref{3dETcls4} cells in $\hK$ might have nonplanar facets (in $d=3$).
We suggested not including these $3$-cells in $\hK$ when we are more interested in its $1$-skeleton.
%But if we want to ensure $\hK$ is indeed a cell complex, we could consider triangulating the nonplanar facets.
%Again, could we do this triangulation while maintaining even degrees for all vertices in $\hK$?
%We may be able to do so if each edge $e \in K$ is shared by the same number of $3$-cells.

\medskip

\noindent {\bfseries Acknowledgment:}
This research was supported in part by an appointment of Gupta to the Oak Ridge National Laboratory ASTRO Program, sponsored by the U.S. Department of Energy and administered by the Oak Ridge Institute for Science and Education.
The authors acknowledge partial funding from the National Science Foundation through grants 1661348 and 1819229.

\bibliographystyle{plain}
%\bibliography{homology,flatnorm,Tool_Path_AM}

\end{document}